\documentclass{article}

\usepackage[main, nonatbib, final]{neurips_2026}

\usepackage[utf8]{inputenc} 
\usepackage[T1]{fontenc}    
\usepackage{booktabs}       
\usepackage{amsmath, amssymb, amsfonts, amsthm}
\usepackage{nicefrac}       
\usepackage{microtype}      
\usepackage{xcolor}         
\usepackage{hyperref}
\hypersetup{
    colorlinks = true,
    citecolor = blue,
    linkcolor = blue,
    urlcolor = blue,
    breaklinks = true,
}
\usepackage{url}

\usepackage{multirow}
\usepackage{algorithm}
\usepackage{subcaption}
\usepackage{enumitem}
\usepackage{algorithmic}
\usepackage{tikz}
\usepackage{mathabx}
\usepackage{mathtools}
\usepackage{comment}
\usepackage{wrapfig}
\usepackage[square,numbers]{natbib}

\newtheorem{theorem}{Theorem}
\newtheorem{remark}{Remark}
\newtheorem{corollary}{Corollary}
\newtheorem{proposition}{Proposition}
\newtheorem{definition}{Definition}

\title{Modeling Quantum Neural Network Gradient With Reinforcement Learning}

\author{
  Nhan Trong Luu$^{1, 3}$\thanks{Corresponding author. Email: \texttt{luutn@ctu.edu.vn}.}\;\;, Duong Trung Luu$^{2}$, Pham Ngoc Nam$^{3}$ and Truong Cong Thang$^{4}$\\
  $^{1}$College of Information and Communication Technology, Can Tho University, Can Tho, Vietnam\\
  $^{2}$Center for Digital Transformation and Communications, Can Tho University, Can Tho, Vietnam\\
  $^{3}$College of Engineering and Computer Science, VinUniversity, Hanoi, Vietnam\\
  $^{4}$Department of Computer Science and Engineering, The University of Aizu, Fukushima, Japan\\
}

\begin{document}

\maketitle

\begin{abstract}
Training quantum neural networks (QNNs) on near-term hardware remains hampered by two compounding difficulties: the exponential vanishing of gradient variance known as the barren plateau, and the $\mathcal{O}(L \cdot 2^n)$ time and memory cost of differentiating through an $n$-qubit, $L$-layer circuit. We propose RLQ-Grad, a reinforcement-learning-based optimizer in which a classical policy $\pi_\phi$ (a spectrally-normalized PPO agent) learns to propose parameter updates directly, conditioned on the QNN's current parameters, loss, accuracy, and previous update. Because the surrogate gradient is emitted by a classical network rather than obtained by differentiating through the unitary $U(\theta)$, its variance is not constrained by the barren plateau concentration bound, and its cost scales with the number of trainable parameters rather than the Hilbert-space dimension. We prove these properties formally and verify them on a hardware-efficient ansatz across four supervised benchmarks with up to $n=20$ qubits. RLQ-Grad preserves a near-flat gradient-variance curve where backpropagation, parameter-shift, and adjoint differentiation decay by 1 to 2 orders of magnitude. Accounting for the full training pipeline (PPO rollouts, actor-critic updates, and optimizer states), RLQ-Grad needs under 2 MB of memory and runs $2490\times$, $7876\times$, and $673\times$ faster per iteration than these three methods at $n=20$. It improves top-1 accuracy by up to $+10\%$ over gradient-based baselines on circuits of up to 12 qubits, and matches dedicated barren plateau mitigation methods on CIFAR-10 at 14 to 20 qubits, where evolutionary and gradient-free optimizers collapse to chance.
\end{abstract}

\section{Introduction}

The emergence of Noisy Intermediate-Scale Quantum (NISQ) devices has established variational quantum algorithms (VQAs) as a leading paradigm for achieving practical quantum advantage~\cite{verdon2019learning, kolle2024study, ostaszewski2021reinforcement}. These hybrid quantum-classical frameworks, which include Quantum Neural Networks (QNNs), rely on the iterative optimization of parameterized quantum circuits (PQCs) to minimize a problem-specific cost function~\cite{verdon2019learning, fodera2024reinforcement, kundu2025tensorrl}. However, the optimization landscape of QNNs is notoriously challenging due to non-convexity, the presence of low-quality local optima, the "barren plateau" phenomenon~\cite{mcclean2018barren}, and the high computational cost of circuit evaluations~\cite{khairy2020learning, verdon2019learning, khairy2019reinforcement}. To address these hurdles, a burgeoning body of research has explored the application of reinforcement learning (RL) to automate and enhance the optimization of both the variational parameters and the underlying circuit architectures.

Initial efforts in this domain focused on utilizing RL to find optimal variational parameters for fixed-structure circuits, such as the Quantum Approximate Optimization Algorithm (QAOA)~\cite{khairy2019reinforcement, khairy2020learning} which formulated parameter optimization as a learning task, where objective is to train RL policy networks on small problem instances to exploit geometrical regularities that generalize to larger, unseen instances. Similarly, \cite{verdon2019learning} introduced a "learning to learn" framework using classical recurrent neural networks (RNNs) as black-box controllers to propose parameter updates for QAOA and Variational Quantum Eigensolvers (VQE), eventually highlights the potential for meta-learning to identify problem-class-specific heuristics that significantly reduce the number of optimization iterations required. Further advancing this theme, \cite{wauters2020reinforcement} proposed an RL scheme for feedback quantum control, demonstrating that an agent can learn smooth, optimal adiabatic schedules for quantum Ising chains that are successfully transferable across system sizes.

Parallel to parameter optimization, some research has shifted toward Quantum Architecture Search (QAS) and circuit optimization, where RL agents autonomously design or refine the circuit structure itself. \cite{fosel2021quantum} demonstrated that deep RL could learn generic strategies to reduce circuit depth and gate count by selecting sequences of logically equivalent transformations. For building circuits from scratch, \cite{kuo2021quantum} proposed a DRL framework to search for gate sequences that synthesize specific quantum states without prior physical knowledge. \cite{ostaszewski2021reinforcement} and \cite{kundu2025tensorrl} also further refined QAS for chemistry applications, employing feedback-driven curriculum learning to adapt problem complexity and discover gate-efficient ansatzes for molecules like LiH.

Recent advancements have focused on improving the scalability and robustness of these RL-based methods. \cite{lockwood2021optimizing} evaluated SAC-based optimizers capable of minimizing loss across random circuits of varying sizes and objectives. To tackle the prohibitive simulation costs of large systems, \cite{kundu2025tensorrl} introduced TensorRL-QAS, which warm-starts the RL search using tensor network approximations to narrow the search space and accelerate convergence. Additionally, specialized architectures for machine learning tasks have been explored, such as \cite{dai2024quantum} (RL-QMLAS) aim to perform QAS for QAOA problems and \cite{fodera2024reinforcement} (RLVQC), which discovered novel "Ryz-connected" ansatz families for combinatorial optimization. Finally, \cite{kolle2024study} have systematically investigated the impact of specific techniques (such as data re-uploading~\cite{perez2020data}, output scaling, and metaheuristic strategies~\cite{gharehchopogh2023quantum}) to stabilize training and reduce the parameter count of VQCs within RL environments. 

While RL has driven notable advances in QNN research, its integration into supervised learning for QNNs remains largely underexplored. Early efforts framing RL as an exhaustive parameter search over QNN configurations prove fundamentally ill-suited to larger-scale problems~\cite{lockwood2021optimizing}. Motivated by these limitations, we make the following contributions:
\begin{itemize}[leftmargin=*]
    \item \textbf{A novel RL-based optimizer for QNNs:} We introduce \textbf{RLQ-Grad}, an RL-based optimizer in which an agent learns to accurately approximate QNN gradients during training, rather than performing exhaustive parameter search or act as a complete replacement of traditional optimizer. Across a range of benchmarks, RLQ-Grad consistently outperforms conventional gradient-based optimizers in classification accuracy.
    \item \textbf{Memory reduction and accelerated computation:} Beyond accuracy, RLQ-Grad circumvents the gradient variance decay caused by the barren plateau phenomenon, while replacing the $2^n$ dependence of gradient computation time and memory with a cost that grows only linearly in the number of trainable QNN parameters. RLQ-Grad maintains an approximately flat gradient-variance profile, whereas backpropagation, parameter-shift, and adjoint differentiation exhibit declines of 1 to 2 orders of magnitude. Under full-pipeline accounting (actor, critic, gradient buffers, and Adam states), RLQ-Grad requires $1.95$ MB at $n=20$ qubits (versus $6174$ MB for backpropagation) and achieves per-iteration CPU speedups of $\sim\!2490\times$, $\sim\!7876\times$, and $\sim\!673\times$ over backpropagation, parameter-shift, and adjoint differentiation, respectively, while also transferring directly to GPU execution.
    \item \textbf{Scaling to 20 qubits against stronger baselines:} RLQ-Grad improves top-1 accuracy by up to $+10\%$ over the strongest gradient-based baseline from 2 to 12 qubits. On CIFAR-10 extended to 20 qubits ($P_{\mathrm{QNN}}=600$), it keeps improving where backpropagation stagnates, matches dedicated barren plateau mitigation (layerwise learning~\cite{skolik2021layerwise} and Gaussian initialization~\cite{zhang2022escaping}), and clearly outperforms evolutionary and gradient-free optimizers (CMA-ES, sep-CMA-ES, SPSA), which collapse to chance.
\end{itemize}
To the best of our knowledge, this work is the first to successfully employ a reinforcement learning agent as a gradient-based optimizer for QNNs, reducing the memory complexity of gradient computation to linear scaling while remaining acceleratable using modern GPU.\footnote{The source code for our paper is available on Zenodo at \url{https://doi.org/10.5281/zenodo.22956233}.}

\section{Proposed method}

\subsection{Problem Formulation}

Let $\mathcal{D}=\{(x_i,y_i)\}_{i=1}^{N}$ denote a supervised learning dataset where $x_i \in \mathbb{R}^{d}$ represents the input features and $y_i \in \{0, 1, ... n\}$ denotes the $n$-class target label. We consider training a parameterized quantum neural network QNN defined by a set of parameters $\theta \in \mathbb{R}^{P}$. The goal of QNN supervised training is to minimize an empirical risk function of:
\begin{equation}
\min\mathcal{L}(\theta) = \frac{1}{N}\sum_{i=1}^{N}\ell(f(x_i, \theta),y_i),
\end{equation}
where $f_{\theta}(x)$ denotes the output of the quantum model and $\ell(\cdot)$ is a certain type of supervise learning loss function, and traditional QNN optimization relies on gradient estimation techniques such as backpropagation, adjoint differentiation~\cite{jones2020efficient} or parameter-shift rule~\cite{mitarai2018quantum}. However, these approaches often suffer from high computational cost and gradient vanishing phenomena. Instead of explicitly computing the gradient $\nabla_\theta \mathcal{L}$, we propose to learn a gradient generator using reinforcement learning.

\subsection{QNN selection}

In most of our problem, we consider a HEA-like \cite{leone2024practical} variational quantum circuit composed of $n$ qubits and $L$ layers. Illustration of HEA design used in our work can be found at Figure~\ref{fig:hea}.

\begin{figure*}[t]
    \centerline{\includegraphics[width=\linewidth]{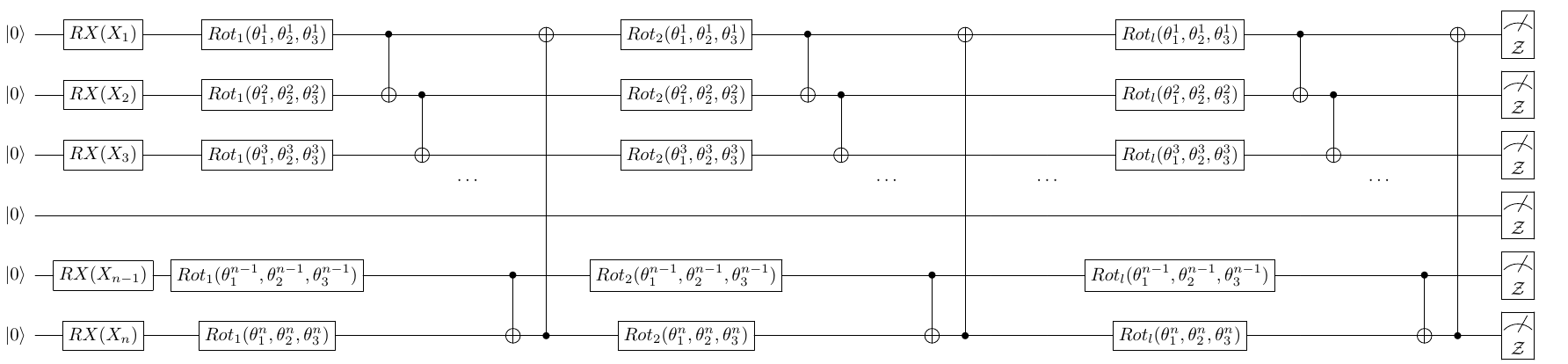}}
    \caption{Illustration of a $n$-qubit $l$-depth HEA circuit utilized in the majority of our experiments.}
    \label{fig:hea}
    \vspace{-0.7\intextsep}
\end{figure*}

In our selected HEA, given classical inputs $x$, projected inputs over a fully connected layer $\mathrm{FC}_{in}$ are encoded through $R_x$ rotation gates as:
\begin{equation}
    U_{\text{enc}}(x) = \prod_{i=1}^{n} R_x(\mathrm{FC}_{in}(x,\mathbf{W}_{in}, \mathbf{b}_{in})_i) \text{ where } \mathrm{FC}(x) = x^\dagger\mathbf{W}+ \mathbf{b}.\\
\end{equation}

Each variational layer consists of parameterized rotations followed by entanglement operations. Let $\theta = \{\theta_{l,q}^{(1)},\theta_{l,q}^{(2)},\theta_{l,q}^{(3)}\}$ denote the parameters of a rotation gate on qubit $q$ in layer $l$, $\mathbf{W} = \{\mathbf{W}_{in}, \mathbf{W}_{out}\}$ and $\mathbf{b} = \{\mathbf{b}_{in}, \mathbf{b}_{out}\}$ for parameter of data encode-decode FC layers. The variational unitary is therefore:
\begin{equation}
U(\theta) =
\prod_{l=1}^{L}
\left(
U_{\text{ent}}
\prod_{q=1}^{n}
Rot(\theta_{l,q})
\right),
\end{equation}
where $Rot(\theta_{l,q})$ represents a general single-qubit rotation and $U_{\text{ent}}$ denotes the entangling operations. Given an input $x$, the circuit prepares the quantum state:
\begin{equation}
|\psi(x,\theta)\rangle
=
U(\theta)U_{\text{enc}}(x)|0\rangle^{\otimes n}.
\end{equation}
The model output is obtained by measuring a Pauli observable $Z$, followed by a fully-connected layer $\textrm{FC}$ that map output measurements to desired class counts:
\begin{equation}
f(x, \theta, \mathbf{W}, \mathbf{b}) = \textrm{FC}_{out}\left( 
\langle \psi(x,\theta) | Z | \psi(x,\theta) \rangle, \mathbf{W}_{out}, \mathbf{b}_{out} \right),
\end{equation}
and resulted output can be utilized as prediction in supervised learning problem\footnote{HEA can be considered a universal approximator, see Appendix \ref{apd:uat} for more information.}. To construct the HEA circuit compatible to our workflow, we utilized both Pennylane~\cite{bergholm2018pennylane} and Torchquantum~\cite{hanruiwang2022quantumnas} at the base of our QNN framework.

While showing promising potential as an alternative to construct predictor in supervise learning, HEA and many other quantum classifier tends to suffer from gradient variance decay, stem from quantum barren plateau~\cite{larocca2025barren} that prohibit them to converge towards an optimal approximation of given problems\footnote{We intentionally selected QNN that are subject to gradient decay to prove that our method is decay free. See Appendix \ref{apd:barren_plateau} for more information regarding definition of plateau.}.

\subsection{RLQ-Grad algorithm}

Instead of computing analytical gradients, we formulate QNN training as a RL problem, where a RL agent interacts with the QNN training process and learns to generate parameter updates and termed our methodology RLQ-Grad. The intuition behind RLQ-Grad is to train a RL agent aim to learn a policy $\pi_{\phi}(a_t|s_t)$ parameterized by $\phi$, which outputs gradient estimates conditioned on the QNN training state. Through interaction with the QNN training environment, the policy gradually learns to generate gradient directions that effectively optimize the quantum model. Pseudocode for RLQ-Grad can be found at Algorithm \ref{alg:rl_qnn}.

\begin{algorithm*}[t]
    \caption{Pseudocode for RLQ-Grad algorithm}
    \label{alg:rl_qnn}
    \begin{algorithmic}[1]
    
    \REQUIRE Train dataset $\mathcal{D}_{\mathrm{train}}=\{(x_i,y_i)\}_{i=1}^{N}$ of $N$ samples with stochastic batch $\mathcal{B} \subset \mathcal{D}_{\mathrm{train}}$, parameterized QNN model $f(x, \theta, \mathbf{W}, \mathbf{b})$, RL policy $\pi_\phi$, total training steps $T$.
    
    \STATE Initialize all QNN parameters $\theta_0$ and previous gradient $g_{-1}$ at zero.
    \STATE Compute initial statistic on first stochastic batch $\mathcal{B}_0$, including loss $\mathbb{E}_{(x,y)\sim \mathcal{B}_0}\left[\mathcal{L}(f(x, \theta_0, \mathbf{W}, \mathbf{b}), y) \right]$ and accuracy $\mathbb{E}_{(x,y)\sim \mathcal{B}_0}\left[acc(f(x, \theta_0, \mathbf{W}_0, \mathbf{b}_0), y)\right]$.
   
    \FOR{$t = 0$ to $T-1$}
        \STATE Construct environment state $s_t = \left[  
        \begin{aligned}
             & \theta_t, \mathbb{E}_{(x,y)\sim \mathcal{B}_t}\left[\mathcal{L}(f(x, \theta_t, \mathbf{W}_t, \mathbf{b}_t), y) \right], \\
             & \mathbb{E}_{(x,y)\sim \mathcal{B}_t}\left[acc(f(x, \theta_t, \mathbf{W}_t, \mathbf{b}_t), y)\right], g_{t-1}
        \end{aligned}\right]$
        \STATE Sample action (gradient proposal) from RL agent $g_t \sim \pi_\phi(\cdot | s_t)$
        \STATE Assign proposed gradient $g_t$ to QNN parameter.
        \STATE Calculate the rest of the gradient graph components and update with backpropagation.
        \STATE Evaluate updated model training loss $\mathbb{E}_{(x,y)\sim \mathcal{B}_t} \left[\mathcal{L}(f(x, \theta_{t+1}, \mathbf{W}_{t+1}, \mathbf{b}_{t+1}), y) \right]$ and training accuracy $\mathbb{E}_{(x,y)\sim \mathcal{B}_t} \left[ acc(f(x, \theta_{t+1}, \mathbf{W}_{t+1}, \mathbf{b}_{t+1}), y)\right]$.
        \STATE Compute reward $\mathbb{E}_{(x,y)\sim \mathcal{B}_t} \left[ acc(f(x, \theta_t, \mathbf{W}_t, \mathbf{b}_t), y) + \frac{1}{\mathcal{L}(f(x, \theta_t, \mathbf{W}_t, \mathbf{b}_t),y) +\epsilon}\right]$.
        \STATE Store transition $(s_t,g_t,r_t,s_{t+1})$ and update RL policy $\pi_\phi$ using collected transitions.
    \ENDFOR
    \RETURN Final trained QNN parameters $\theta$
    \end{algorithmic}
\end{algorithm*}

In most of our experiments, we use Proximal Policy Optimization (PPO)~\cite{schulman2017proximalpolicyoptimizationalgorithms} as our RL agent in our RLQ-Grad algorithm, while also incorporating spectral normalization (SN)~\cite{gogianu2021spectral} improve performance and reducing variance during training. Details on training settings and justification for model selection can be found in Appendix~\ref{apd:misc}.

\subsubsection{State representation}

Given train dataset $\mathcal{D}_{\mathrm{train}}=\{(x_i,y_i)\}_{i=1}^{N}$ of $N$ samples with stochastic batch $\mathcal{B}_t \subset \mathcal{D}_{\mathrm{train}}$, at training step $t$ the environment state is defined as:
\begin{equation}
\begin{aligned}
    & s_t = \left[  \theta_t, \mathbb{E}_{(x,y)\sim \mathcal{B}_t} \left[\mathcal{L}(f(x, \theta_t, \mathbf{W}_t, \mathbf{b}_t), y) \right], \mathbb{E}_{(x,y)\sim \mathcal{B}_t} \left[acc(f(x, \theta_t, \mathbf{W}_t, \mathbf{b}_t), y)\right], g_{t-1}\right],\\
    & g_{t-1} = \mathrm{flatten}\left(\frac{\partial  \,\mathbb{E}_{(x,y)\sim \mathcal{B}_t}\left[\mathcal{L}(f(x, \theta_t, \mathbf{W}_t, \mathbf{b}_t), y) \right]}{\partial\,\mathbb{E}_{(x,y)\sim \mathcal{B}_t}\left[f(x, \theta_t, \mathbf{W}_t, \mathbf{b}_t) \right]}\cdot\frac{\partial\,\mathbb{E}_{(x,y)\sim \mathcal{B}_t}\left[f(x, \theta_t, \mathbf{W}_t, \mathbf{b}_t) \right]}{\partial\,\theta_{t-1}}\right)
\end{aligned}
\end{equation}
where $\theta_t$ are the current QNN parameters, $\mathcal{L}$ is the mean training loss, $acc$ is the mean training accuracy, followed by all flattened trainable parameters' gradient $g_{t-1}$ proposed by the agent at the previous step. Proposed gradient $g_{t}$ is equivalent to RL agent outputs, interpreted as a gradient estimate $a_t = g_t \in \mathbb{R}^{P}$ of $P$ optimizable parameter, and after the agent proposed the gradient $g_t$, it was assigned to QNN parameter gradient. 

Finally, we calculate other components in the backpropagation graph (such as linear weights gradient $\frac{\partial \,\mathbb{E}_{(x,y)\sim \mathcal{B}_t}\left[\mathcal{L}(f(x, \theta_t, \mathbf{W}_t, \mathbf{b}_t), y) \right]}{\partial\,\mathbf{W}_{t}}$ and biases gradient $\frac{\partial  \,\mathbb{E}_{(x,y)\sim \mathcal{B}_t}\left[\mathcal{L}(f(x, \theta_t, \mathbf{W}_t, \mathbf{b}_t), y) \right]}{\partial\,\mathbf{b}_{t}}$, etc.) and update $f(x, \theta, \mathbf{W}, \mathbf{b})$ with Adam optimizer~\cite{kingma2014adam} (learning rate $lr=1e-3$, weight decay of $1e-4$, momentum $\beta = (0.9, 0.999)$) along with cosine annealing scheduler~\cite{loshchilov2016sgdr}.

\subsubsection{Reward function}

We formulate our reward function $r_t$ at timestep $t$ using training loss statistic $\mathcal{L}_t$ and accuracy statistic $acc_t$ over stochastic batch of training set $\mathcal{B}_t \subset \mathcal{D}_{\mathrm{train}}$ of train set $\mathcal{D}_\mathrm{train}$ as:
\begin{equation}
    r_t = \mathbb{E}_{(x,y)\sim \mathcal{B}_t} \left[ acc(f(x, \theta_t, \mathbf{W}_t, \mathbf{b}_t), y) + \frac{1}{\mathcal{L}(f(x, \theta_t, \mathbf{W}_t, \mathbf{b}_t),y) +\epsilon}\right],
\end{equation}
where $\epsilon$ is a small constant to prevent zero division. The goal to encourages the agent to produce gradients that simultaneously reduce training loss and improve predictive performance.

\section{Results}

\subsection{Surrogate gradient advantages of RLQ-Grad}

\begin{wrapfigure}{l}{0.6\textwidth}
\vspace{-\intextsep}
\centering
\caption{Comparison of gradient variance of the first trainable parameter in the HEA $\mathrm{Var}(\langle \partial\, \theta^{(1)}_{1,1} \, E \rangle)$ of gradient computation strategies for a depth-$L$, $n$-qubit QNN with $P_{\mathrm{QNN}}$ trainable parameters. Variance results for RLQ-Grad are obtained from post training agent on BC dataset on Section \ref{sec:main_res}.}
\includegraphics[width=\linewidth]{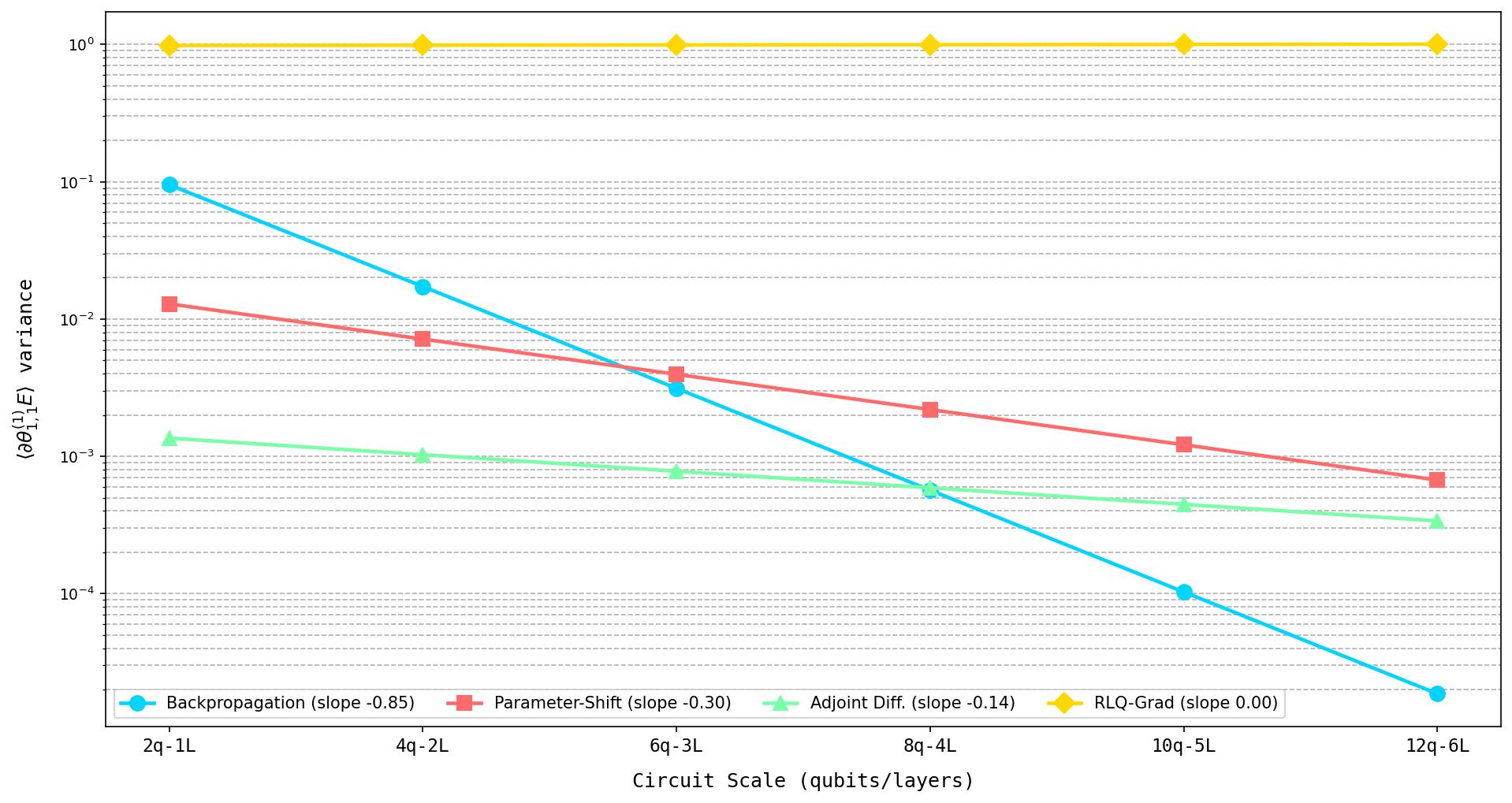}
\label{fig:plateau_mea}
\vspace{-0.8\intextsep}
\end{wrapfigure}

We assess the advantages of surrogate gradient produced by RLQ-Grad in light of the recently established unified QNN loss landscape theory of \cite{anschuetz2024unified}. We claim that proposed RLQ-Grad algorithm can generate gradient that free from gradient decay caused by quantum barren plateau~\cite{larocca2025barren} over the following theorem:
\begin{theorem}[Gradient variance decay free]\label{thm:bypass}
We established some notations as follow: $\mathcal{A} \cong \bigoplus_\alpha \mathcal{A}_\alpha$ denotes the Jordan algebra generated by $\{U_\theta^\dagger O U_\theta\}_\theta$, $\mathrm{Aut}(\mathcal{A}_\alpha)$ its automorphism group, and $r_\alpha = \lfloor \mathrm{Tr}_\alpha(O_\alpha)^2/\mathrm{Tr}_\alpha(O_\alpha^2) \rceil$ the degrees of freedom of the $\alpha$-th simple component. Let $U(\theta)$ be a QNN whose associated Jordan algebra $\mathcal{A}$ satisfies:
\begin{equation}
    \mathrm{Var}_\theta[\ell(\rho;\theta)]
    = \sum_\alpha
    \frac{\mathrm{Tr}(O_\alpha)^2\,\mathrm{Tr}(\rho^\alpha)^2}
         {\mathrm{dim}_\mathbb{R}(\mathrm{Aut}(\mathcal{A}_\alpha))}
    \;\in\; \mathcal{O}\!\left(\mathrm{poly}(\log N)^{-1}\right),
    \label{eq:jaws_bp}
\end{equation}
i.e., the circuit is in a barren plateau\footnote{Original derivation can be found in Definition~7 of \cite{anschuetz2024unified}.}. Under RLQ-Grad, the parameter update at step $t$ is $\theta_{t+1} = \theta_t - \eta\,g_t$ where $g_t\sim\pi_\phi(\cdot\mid s_t)$ is the RL agent's output. Since $g_t$ is not computed by differentiating through
$U(\theta)$, the vanishing of Equation~\ref{eq:jaws_bp} places no constraint on $g_t$.
The variance decay is therefore structurally bypassed.
\end{theorem}

\begin{proof}
    See Appendix~\ref{proof:thm:bypass}.
\end{proof}

To provide a more practical and comprehensive evaluation of the decay resistance of our method with respect to (w.r.t.) well-established gradient-based approaches, Figure~\ref{fig:plateau_mea} reports the gradient variance of the first trainable parameter in the HEA $\mathrm{Var}(\langle \partial\, \theta^{(1)}_{1,1} \, E \rangle)$, as a function of both qubit count and circuit depth\footnote{Denoted in the format $n\mathrm{q}L\mathrm{d}$ for an $n$-qubit, $L$-layer HEA. This notation is used universally in our paper.}. The comparison includes several gradient estimation techniques, namely traditional backpropagation, parameter-shift~\cite{mitarai2018quantum}, adjoint differentiation~\cite{jones2020efficient}, and RLQ-Grad (ours). As illustrated in Figure~\ref{fig:plateau_mea}, our method exhibits robustness against decay, indicating a clear advantage over the baseline approaches.

Still, even when $g_t$ is free of variance decay, does the agent's own training suffer from an analogous pathology? We can analyze this question by forming the follow theorem:
\begin{theorem}[Plateau unrelated with agent optimization]\label{thm:agent_grad}
The RL policy $\pi_\phi$ is a classical neural network (e.g.\ an MLP). Its parameter
gradient $\nabla_\phi \mathcal{L}_{\mathrm{RL}}$ is computed entirely by classical
backpropagation and is therefore not subject to the variance decay phenomenon.
\end{theorem}
Proof for this theorem can be found at Appendix~\ref{proof:thm:agent_grad}. Taken together, Theorems~\ref{thm:bypass} and~\ref{thm:agent_grad} formalize a key advantage of RLQ-Grad: it decouples the optimization signal from the quantum circuit’s intrinsic geometry while retaining a fully trainable classical learning mechanism. By construction, the surrogate gradient $g_t$ is not constrained by the variance collapse characterized in Equation~\ref{eq:jaws_bp}, allowing meaningful updates even in regimes where standard gradient-based methods fail. At the same time, the learning dynamics of the policy $\pi_\phi$ remain governed by well-behaved classical optimization, ensuring that no analogous decay emerges on the agent side. Other properties regarding training conditions and environment setting effects can be found at Appendix~\ref{apd:trainability} and \ref{thm:reward_affect}.

\subsubsection{Resource advantages}\label{sec:resource}

\begin{figure*}[t!]
    \centering
    \includegraphics[width=\textwidth]{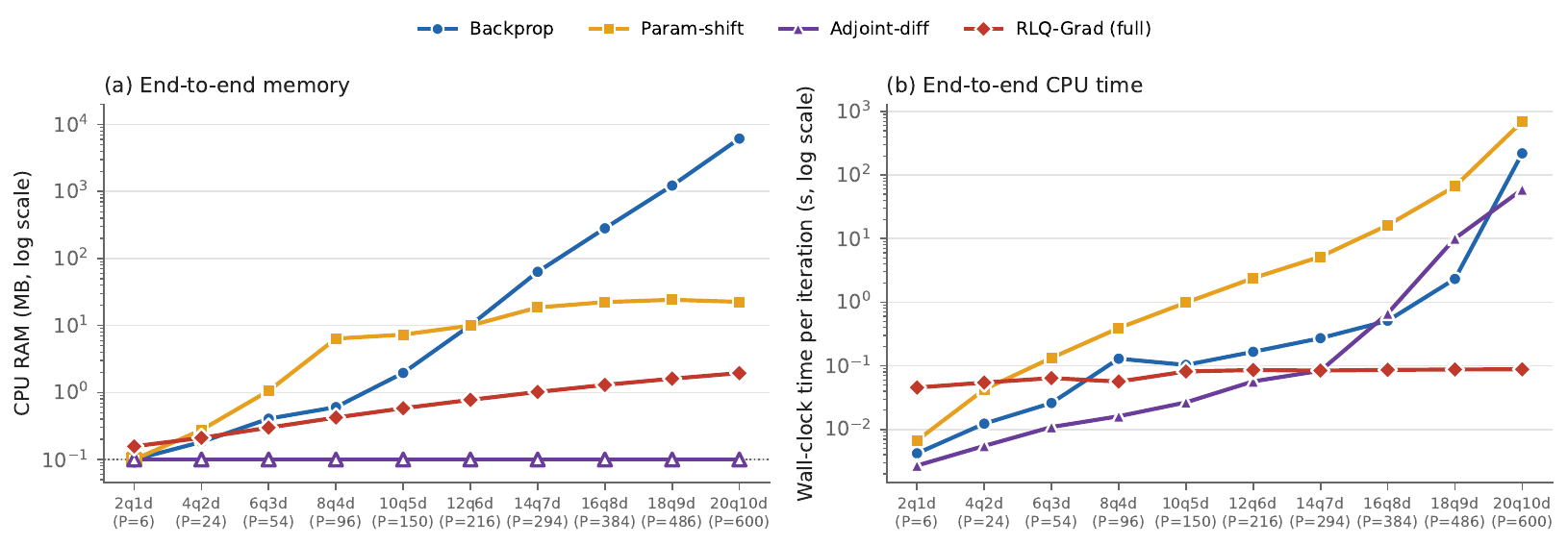}
    \caption{End-to-end memory (a, CPU RAM) and per-iteration CPU wall-clock time (b) of QNN gradient computation methods from 2 to 20 qubits (log scale, mean over 12 trials after 4 warm-up iterations). RLQ-Grad costs cover the full training pipeline: PPO rollout, actor and critic updates, gradient buffers, and Adam optimizer states. Hollow markers denote measurements below the 0.1 MB resolution. Exact values are listed in Appendix~\ref{apd:resource_tables}.}
    \label{fig:membench}
    \vspace{-0.7\intextsep}
\end{figure*}

\begin{wraptable}{l}{0.6\textwidth}
\vspace{-\intextsep}
\centering
\caption{Theoretical complexity of gradient computation strategies for a depth-$L$, $n$-qubit QNN with $P_{\mathrm{QNN}}$ trainable parameters. Forward pass complexity is identical for all methods and therefore omitted. For an MLP agent of depth $L_{\mathrm{MLP}}$ and width $d$, $P_{\mathrm{RL}} = \mathcal{O}(L_{\mathrm{MLP}} d^2 + d\,P_{\mathrm{QNN}})$ counts actor and critic parameters.}
\resizebox{\linewidth}{!}{
\begin{tabular}{@{}lcc@{}}
\toprule
\bf Method & \bf \shortstack{Time Complexity\\$\mathcal{O}_{\mathrm{time}}$} & \bf \shortstack{Memory Complexity\\$\mathcal{O}_{\mathrm{memory}}$} \\
\midrule
Backpropagation~\cite{bergholm2018pennylane} & $L \cdot 2^n$                       & $L \cdot 2^n$ \\
Parameter-shift~\cite{mitarai2018quantum} & $P_{\mathrm{QNN}} \cdot L \cdot 2^n$ & $P_{\mathrm{QNN}} \cdot2^n$ \\
Adjoint differentiation~\cite{jones2020efficient} & $P_{\mathrm{QNN}} \cdot 2^n$ & $1$\\
\midrule
\bf RLQ-Grad (actor inference) & $L_{\mathrm{MLP}} d^2 + d\,P_{\mathrm{QNN}}$ & $d + P_{\mathrm{QNN}}$ \\
\bf RLQ-Grad (full pipeline) & $P_{\mathrm{RL}}$ & $4\,P_{\mathrm{RL}}$ \\
\bottomrule
\end{tabular}}
\label{tab:complexity}
\vspace{-0.7\intextsep}
\end{wraptable}
 
To quantify how our method performs with respect to (w.r.t.) well-known gradient-based methods, Table~\ref{tab:complexity} summarises the asymptotic cost of each gradient computation scheme of Figure~\ref{fig:plateau_mea}. Backpropagation stores all $L$ intermediate statevectors, incurring $\mathcal{O}(L \cdot 2^n)$ in both time and memory. Adjoint differentiation costs $\mathcal{O}_{\mathrm{time}}(P_{\mathrm{QNN}} \cdot 2^n)$ but reconstructs intermediate states on-the-fly via unitary inversion, so beyond the statevector already held by the forward pass it needs only a constant number of buffers, i.e.\ $\mathcal{O}_{\mathrm{memory}}(1)$ additional memory. The parameter-shift rule evaluates two full forward passes per parameter, giving $\mathcal{O}_{\mathrm{time}}(P_{\mathrm{QNN}} \cdot L \cdot 2^n)$ at $\mathcal{O}_{\mathrm{memory}}(P_{\mathrm{QNN}} \cdot 2^n)$, making it the least time/memory-efficient of the three for deep or wide circuits.

RLQ-Grad replaces explicit gradient propagation through the quantum circuit with a learned surrogate estimator. Since the agent operates entirely in the classical domain, its cost depends on its own parameter count $P_{\mathrm{RL}}$, not on the Hilbert space dimension $2^n$. Because both the state $s_t$ and the action $g_t$ have dimension $\mathcal{O}(P_{\mathrm{QNN}})$, only the input and output layers of the MLP depend on the circuit size, while the hidden layers keep a fixed width $d$. Proposing a gradient (actor inference) therefore costs $\mathcal{O}(L_{\mathrm{MLP}} d^2 + d\,P_{\mathrm{QNN}})$ time. A complete training step must also hold the critic, the gradient buffers, and the two Adam moment estimates, so its memory is $\approx 4P_{\mathrm{RL}}$. Both costs grow linearly in $P_{\mathrm{QNN}}$, hence polynomially in $n$, and are independent of $2^n$, shifting the computational burden from quantum simulation to classical function approximation. Full derivations are given in Appendix~\ref{apd:complexity}.

Figure~\ref{fig:membench} reports the measured end-to-end cost for $n=2$ to $20$. Backpropagation grows from $0.0042$s to $219.41$s per iteration, following its $\mathcal{O}_{\mathrm{time}}(L \cdot 2^n)$ complexity; adjoint differentiation grows from $0.0027$s to $59.32$s, and parameter-shift, which requires $2P_{\mathrm{QNN}}$ circuit evaluations, from $0.0067$s to $693.96$s. The full RLQ-Grad pipeline, roughly $50\times$ the cost of an actor-only forward pass, instead stays between $0.0455$s and $0.0881$s across all configurations. Its fixed overhead makes it slower than analytic differentiation on the smallest circuits, but it overtakes parameter-shift from 6q3d, backpropagation from 8q4d, and adjoint differentiation from 16q8d (the two are tied at 14q7d). At $n=20$, RLQ-Grad is $\sim\!2490\times$, $\sim\!7876\times$, and $\sim\!673\times$ faster than backpropagation, parameter-shift, and adjoint differentiation, respectively. RLQ-Grad also runs natively on GPU, whereas parameter-shift and adjoint differentiation lack efficient GPU kernels in current QNN frameworks; on an RTX 3090 it is faster than GPU backpropagation at every scale ($1.03\times$ to $1.22\times$, Appendix~\ref{apd:resource_tables}).

Memory consumption further delineates the scalability of these methods. Backpropagation grows exponentially from below $0.1$ MB to $6174$ MB at $n=20$, while parameter-shift saturates around $20$ MB. Adjoint differentiation stays below $0.1$ MB at every scale, consistent with its $\mathcal{O}(1)$ additional memory, and remains the most memory-efficient method. The full RLQ-Grad pipeline grows from $161.22$ KB to $1.95$ MB, i.e.\ by $\approx 3$ KB per QNN parameter, confirming the predicted linear dependence on $P_{\mathrm{QNN}}$. This footprint decomposes into actor and critic parameters ($\approx 26\%$), gradient buffers ($\approx 26\%$), and Adam optimizer states ($\approx 52\%$); at 12q6d, for example, the total of $796.14$ KB is roughly $4\times$ the actor-only $102.28$ KB. Even under this full accounting, RLQ-Grad scales strictly better than backpropagation and parameter-shift, using $\sim\!3166\times$ and $\sim\!12\times$ less memory at $n=20$.

\subsection{Benchmarking against well-known quantum optimization algorithms}\label{sec:main_res}

\begin{table*}[t!]
    \caption{Comparison of top 1 validation accuracy (in percentage, averaged over 4 runs) of gradient computation strategies for a $n$-qubit, $L$-depth HEA. Best performance per variant is denoted in bold and second best is underlined.}
    \centering
    \resizebox{\textwidth}{!}{
        \begin{tabular}{@{}ll*{6}{c}@{}}
        \toprule
        \multirow{3}{*}{\bf Datasets} & \multirow{3}{*}{\bf Algorithms} & \multicolumn{6}{c}{\bf Top 1 validation accuracy per configuration}\\ 
        \cmidrule{3-8}
        & & \shortstack{2q1d\\(Param = 6)} & \shortstack{4q2d\\(Param = 24)} & \shortstack{6q3d\\(Param = 54)} & \shortstack{8q4d\\(Param = 96)} & \shortstack{10q5d\\(Param = 150)} & \shortstack{12q6d\\(Param = 216)}\\
        \midrule
        \multirow{4}{*}{\shortstack{BC}} & Backpropagation~\cite{bergholm2018pennylane} & $80.62\pm10.64$ & $\underline{90.86\pm1.998}$ & $\underline{91.88\pm2.057}$ & $\underline{91.88\pm1.126}$ & $91.96\pm2.402$ & $\underline{93.15\pm2.942}$\\
        & Parameter-shift~\cite{mitarai2018quantum} & $\underline{85.71\pm7.082}$ & $89.77\pm1.316$ & $89.60\pm5.526$ & $89.77\pm2.010$ & $91.99\pm5.385$ & $92.11\pm5.526$\\
        & Adjoint-differentiation~\cite{jones2020efficient} & $83.01\pm11.64$ & $84.93\pm14.82$ & $85.04\pm3.689$ & $87.32\pm8.476$ & $\underline{92.75\pm10.28}$ & $93.02\pm8.926$\\
        & \bf RLQ-Grad (Ours) &  $\mathbf{95.52\pm0.080}$ &  $\mathbf{95.20\pm2.379}$ &  $\mathbf{95.79\pm3.418}$ &  $\mathbf{95.96\pm0.380}$ &  $\mathbf{96.18\pm4.178}$ &  $\mathbf{96.58\pm5.612}$\\
        \midrule
        \multirow{4}{*}{\shortstack{MNIST}} & Backpropagation~\cite{bergholm2018pennylane} & $\underline{59.24\pm3.171}$ & $78.52\pm1.371$ & $\underline{88.81\pm0.434}$ & $\underline{91.43\pm0.624}$ & $91.74\pm0.305$ & $\underline{92.48\pm0.225}$\\
        & Parameter-shift~\cite{mitarai2018quantum} & $56.25\pm0.419$  & $\underline{80.59\pm1.445}$ & $85.35\pm0.447$ & $87.14\pm0.406$ & $88.43\pm0.181$ & $88.52\pm0.279$\\
        & Adjoint-differentiation~\cite{jones2020efficient} &  $58.16\pm2.086$ & $79.11\pm1.902$ & $87.42\pm0.613$ & $90.85\pm0.771$ & $\underline{92.19\pm0.151}$ & $91.96\pm0.412$\\
        & \bf RLQ-Grad (Ours) & $\mathbf{61.36\pm2.216}$ & $\mathbf{83.12\pm0.375}$  & $\mathbf{90.76\pm0.462}$ & $\mathbf{91.78\pm0.218}$ & $\mathbf{93.63\pm0.114}$ & $\mathbf{95.27\pm0.297}$ \\
        \midrule
        \multirow{4}{*}{\shortstack{F-MNIST}} & Backpropagation~\cite{bergholm2018pennylane} & $\underline{72.87\pm0.312}$& $74.77\pm0.573$ & $\underline{82.34\pm0.441}$ & $81.36\pm0.407$ & $84.62\pm0.271$ & $\underline{85.72\pm0.198}$\\
        & Parameter-shift~\cite{mitarai2018quantum} & $68.84\pm0.154$& $72.63\pm0.601$ & $79.75\pm0.518$ & $\underline{84.12\pm0.336}$ & $82.41\pm0.295$ & $82.98\pm0.266$\\
        & Adjoint-differentiation~\cite{jones2020efficient} & $73.23\pm0.576$& $\underline{75.96\pm0.528}$ & $81.28\pm0.462$ & $83.54\pm0.388$ & $\underline{85.08\pm0.221}$ & $85.13\pm0.243$\\
        & \bf RLQ-Grad (Ours) & $\mathbf{73.64\pm0.415}$ & $\mathbf{76.43\pm0.442}$ & $\mathbf{83.15\pm0.396}$ & $\mathbf{85.87\pm0.285}$ & $\mathbf{86.94\pm0.204}$ & $\mathbf{87.62\pm0.181}$\\
        \midrule
        \multirow{4}{*}{\shortstack{CIFAR10}} & Backpropagation~\cite{bergholm2018pennylane} & $\underline{24.69\pm1.429}$ & $\underline{33.68\pm0.990}$  & $35.02\pm0.884$ & $\underline{39.78\pm0.461}$  & $\underline{40.36\pm0.375}$ & $38.05\pm0.601$\\
        & Parameter-shift~\cite{mitarai2018quantum} & $26.34\pm1.585$ & $32.11\pm1.102$ & $\underline{37.34\pm0.475}$ & $37.45\pm0.722$ & $38.12\pm0.653$ & $\underline{40.08\pm0.321}$\\
        & Adjoint-differentiation~\cite{jones2020efficient} & $25.06\pm0.632$ & $31.87\pm0.955$ & $34.76\pm0.801$ & $36.98\pm0.690$ & $37.65\pm0.588$ & $37.52\pm0.544$\\
        & \bf RLQ-Grad (Ours) & $\mathbf{27.17\pm0.782}$ & $\mathbf{34.84\pm0.744}$ & $\mathbf{37.91\pm0.533}$ & $\mathbf{40.72\pm0.428}$ & $\mathbf{41.63\pm0.367}$ & $\mathbf{41.15\pm0.352}$\\
        \bottomrule
        \end{tabular}
    }
    \label{tab:cross-dataset-eval}
    \vspace{-0.8\intextsep}
\end{table*}

Table~\ref{tab:cross-dataset-eval} and Figure~\ref{fig:acc_main} (Appendix~\ref{apd:acc_curves}) present a comprehensive comparison of gradient computation strategies across multiple datasets for 200 epochs, including: BreastCancer (BC)~\cite{breast_cancer_14}, Fashion-MNIST (F-MNIST)~\cite{xiao2017fashion}, MNIST~\cite{deng2012mnist} and CIFAR10~\cite{krizhevsky2009learning}, along with HEA configurations across gradient computation methods. Except for model trained on BC dataset with use a batch size of 128 due to small sample size, all dataset used a batch size of 256. We also use a fixed random seed sampler for dataset without predetermined split such as BC, and reduce input dimensionality with PCA (Appendix~\ref{apd:pca}).

Overall, RLQ-Grad consistently achieves the best performance across all datasets and configurations, demonstrating both strong accuracy and scalability as the number of qubits and circuit depth increase. Across all datasets, accuracy generally improves as the circuit size increases (from $2q1d$ to $12q6d$), indicating that higher parameter counts enhance representational capacity. Conventional methods such as backpropagation, parameter-shift, and adjoint differentiation exhibit diminishing returns or instability at higher configurations, whereas RLQ-Grad maintains steady improvements.

On BC, RLQ-Grad outperforms all baselines at every configuration, reaching up to $96.58\%$ with a margin of $+2\%$ to $+10\%$ over the second-best method. On MNIST, the gap widens with circuit size, and RLQ-Grad surpasses the competitive backpropagation and adjoint differentiation by up to $\sim\!3\%$ at 12q6d. On F-MNIST, RLQ-Grad again achieves the highest accuracy, reaching $87.62\%$. On CIFAR10, overall accuracies are lower due to dataset complexity, but RLQ-Grad still gives the best results at every configuration, with more modest ($\sim\!2\%$ to $4\%$) but consistent gains.

Backpropagation generally performs well and often ranks second, especially in simpler datasets (BC, MNIST). Parameter-shift shows competitive performance at smaller scales but struggles to scale efficiently. Adjoint differentiation exhibits higher variance and inconsistent ranking, particularly in deeper circuits. In contrast, RLQ-Grad not only achieves the highest accuracy but also maintains relatively stable variance across runs. A key advantage of RLQ-Grad is its scalability, where performance continues to improve with increasing circuit size without the degradation or saturation observed in other methods. Additionally, its relatively low standard deviation (especially on MNIST, F-MNIST, and CIFAR10) indicates stable training dynamics.

\subsection{Scaling to 20 qubits and comparison with stronger baselines}\label{sec:scale}

To probe the regime where barren plateaus become acute and where the action space of the agent is largest, we extend CIFAR-10 to $n = 14, 16, 18, 20$ qubits ($P_{\mathrm{QNN}} = 294, 384, 486, 600$) and add two families of baselines: evolutionary and gradient-free optimizers (CMA-ES~\cite{hansen2016cma}, sep-CMA-ES~\cite{omidvar2010comparative}, and SPSA~\cite{spall1998implementation}), and QNN-specific barren plateau mitigation (layerwise learning~\cite{skolik2021layerwise} and Gaussian initialization~\cite{zhang2022escaping}). Results are reported in Table~\ref{tab:scale}.

\begin{table*}[t!]
    \caption{CIFAR-10 top-1 validation accuracy (in percentage, mean $\pm$ std over 3 runs) at larger circuit scales. Best performance per configuration is denoted in bold and second best is underlined.}
    \centering
    \resizebox{0.9\linewidth}{!}{
        \begin{tabular}{@{}llcccc@{}}
        \toprule
        \bf \multirow{2}{*}{Type} & \bf \multirow{2}{*}{Method} & \multicolumn{4}{c}{\bf Top 1 validation accuracy per configuration}\\ 
        \cmidrule{3-6}
        & & \shortstack{14q7d\\(Param = 294)} & \shortstack{16q8d\\(Param = 384)} & \shortstack{18q9d\\(Param = 486)} & \shortstack{20q10d\\(Param = 600)}\\
        \midrule
        Gradient-based & Backpropagation~\cite{bergholm2018pennylane} & $38.72\pm3.21$ & $40.45\pm2.67$ & $39.86\pm3.48$ & $40.13\pm2.91$\\
        \midrule
        \multirow{3}{*}{Gradient-free} & CMA-ES~\cite{hansen2016cma} & $10.10\pm0.14$ & $10.12\pm0.21$ & $10.11\pm0.20$ & $10.13\pm0.15$\\
        & sep-CMA-ES~\cite{omidvar2010comparative} & $10.09\pm0.23$ & $10.14\pm0.17$ & $10.11\pm0.26$ & $10.15\pm0.19$\\
        & SPSA~\cite{spall1998implementation} & $10.08\pm0.27$ & $10.15\pm0.18$ & $10.11\pm0.32$ & $10.14\pm0.24$\\
        \midrule
        \multirow{2}{*}{BP mitigation} & Layerwise~\cite{skolik2021layerwise} & $44.73\pm3.21$ & $\underline{48.58\pm2.74}$ & $\underline{50.92\pm3.47}$ & $\mathbf{53.16\pm3.65}$\\
        & Gaussian Init~\cite{zhang2022escaping} & $43.92\pm3.15$ & $46.87\pm2.38$ & $49.76\pm3.02$ & $52.41\pm2.71$\\
        \midrule
        \textbf{Learnable (ours)} & \bf RLQ-Grad & $\mathbf{45.86\pm3.84}$ & $\mathbf{48.64\pm2.46}$ & $\mathbf{51.29\pm3.16}$ & $\underline{53.02\pm2.83}$\\
        \bottomrule
        \end{tabular}
    }
    \label{tab:scale}
    \vspace{-0.8\intextsep}
\end{table*}

\subsubsection{Action-space scalability and non-stationarity}

RLQ-Grad remains stable and improves monotonically up to $P_{\mathrm{QNN}} = 600$. The action dimensionality enters the agent only through its input and output layers, while the hidden layers keep a fixed width, so the agent grows linearly and its training difficulty grows sub-linearly in $P_{\mathrm{QNN}}$. The environment is non-stationary because the QNN changes under the agent's own actions; PPO with SN handles this through (i) on-policy rollouts, which avoid off-policy staleness, (ii) the clipped surrogate objective, which bounds the divergence between successive policies, and (iii) SN on the critic, which bounds the propagation of TD errors (Appendix~\ref{sec:sn_norm_vf}). The standard deviations of $2.46$ to $3.84$ in Table~\ref{tab:scale}, comparable to those of backpropagation and the mitigation baselines, confirm that training remains stable at this scale.

\subsubsection{Why reinforcement learning rather than evolutionary search}

CMA-ES, sep-CMA-ES, and SPSA collapse to chance ($\sim\!10\%$) at every scale. This is consistent with \cite{arrasmith2021effect}: gradient-free optimizers cannot escape the exponential shot-budget scaling of high-dimensional barren plateaus, and the $\mathcal{O}(P_{\mathrm{QNN}}^2)$ covariance estimation of CMA-ES becomes prohibitive for $P_{\mathrm{QNN}} \geq 300$. RLQ-Grad avoids both issues because its amortized policy learns a state-conditional update rule instead of searching over a covariance model. Layerwise learning and Gaussian initialization reach accuracy comparable to RLQ-Grad, but they remain backpropagation variants that inherit its $\mathcal{O}(L \cdot 2^n)$ time and memory (Figure~\ref{fig:membench}). RLQ-Grad matches their accuracy while retaining its computational advantages, which positions it as complementary to landscape-shaping approaches rather than as a replacement.

\subsubsection{Are the benchmarks in the barren plateau limited regime?}

The range $n = 2$ to $12$ of Table~\ref{tab:cross-dataset-eval} is best characterized as an \emph{early-onset} rather than a \emph{fully saturated} barren plateau regime. Figure~\ref{fig:plateau_mea} shows that gradient variance already decays with slopes of $-0.85$ (backpropagation), $-0.30$ (parameter-shift), and $-0.14$ (adjoint differentiation) over this range, versus $\approx 0$ for RLQ-Grad, but the decay has not yet dominated final accuracy; this is expected, since barren plateaus become acute at $n \gtrsim 20$ for global observables in generic HEAs~\cite{larocca2025barren}. The extended experiments show the pattern predicted if barren plateaus are the operative bottleneck. Backpropagation gains only $\sim\!2$ pp from 12q6d to 20q10d ($38.05\% \to 40.13\%$) although $P_{\mathrm{QNN}}$ nearly triples ($216 \to 600$), compared with $\sim\!13$ pp from 2q1d to 12q6d ($24.69\% \to 38.05\%$); its rate of improvement collapses where the barren plateau is expected to activate. Over the same range, layerwise learning and Gaussian initialization gain $\sim\!8.5$ pp, and RLQ-Grad follows the same trajectory ($45.86\% \to 53.02\%$).

These observations also separate learned-optimizer dynamics from barren plateau mitigation. If RLQ-Grad won through its optimizer dynamics alone, it should outperform mitigation methods that lack such dynamics; instead, the three methods stay within $2$ pp of each other from 14q7d to 20q10d, which suggests that all three succeed through the same mechanism, namely avoiding the variance collapse. Conversely, if the advantage were purely a matter of avoiding circuit gradients, any gradient-free method would succeed, which CMA-ES and SPSA rule out. Gradient-freeness (Theorem~\ref{thm:bypass}) is therefore necessary but not sufficient, and the contribution of RLQ-Grad is to combine it with amortized policy learning that can handle a high-dimensional action space.

\section{Limitations and discussions}

Our analysis identifies three concrete regimes in which RLQ-Grad does not offer the advantage suggested by Theorems \ref{thm:bypass} and \ref{thm:agent_grad}, and two practical caveats that temper the scope of the reported gains.

\paragraph{Reward concentration in deep barren-plateau regimes.} Theorem~\ref{thm:bypass} establishes that the surrogate gradient $g_t$ is not constrained by the Jordan-algebraic variance bound because $g_t$ is emitted by a classical network rather than computed by differentiating through $U(\theta)$. However, as we show in Proposition~\ref{prop:reward_info} (Appendix~\ref{thm:reward_affect}), the reward $r_t = \mathrm{acc}_t + (\ell_t + \epsilon)^{-1}$ is itself a function of the loss $\ell_t$, and Wishart-process description implies that $\ell_t$ concentrates exponentially whenever the circuit is on a barren plateau~\cite{anschuetz2024unified}. In that regime the reward differences used to drive policy gradients vanish at the same rate as conventional cost-function differences, echoing the result of \cite{arrasmith2021effect} that no gradient-free optimizer (including surrogate-gradient models) can escape the exponential shot-budget scaling of a barren plateau. RLQ-Grad therefore delivers a genuine landscape advantage only when the Jordan algebra $\mathcal{A}$ leaves the cost non-flat (Condition 1 of Theorem~\ref{thm:trainability}); when the cost is truly exponentially flat, RLQ-Grad retains only its computational advantage.

\paragraph{Poor local minima are not bypassed.} Theorem~\ref{thm:local_minima} (Appendix~\ref{apd:local_min_lim}) shows that, independent of gradient quality, the density of good local minima is exponentially suppressed whenever the overparameterization condition $p \ge \max_\alpha \beta_\alpha r_\alpha$ fails. Because RLQ-Grad still converges to a local minimum of the empirical loss, it inherits this obstruction. Figure~\ref{fig:ll} makes this visible: on the underparameterized F-MNIST configuration the loss surface fragments into sharp basins for every method, RLQ-Grad included.

\paragraph{Input-gradient decay propagates to classical pre-processing.} The hybrid architecture we use funnels classical features through an FC layer $\mathrm{FC}_{\text{in}}$ before the $R_x$ encoding. Theorem \ref{thm:ibp} and Corollary \ref{col:hybrid} (Appendix \ref{apd:input_grad}) show that whenever $U(\theta)$ forms (at least) a 2-design, the input-space Jacobian $\mathbb{E}_\theta[\|\nabla_x \mathcal{L}\|^2] \le C n / 2^n$ collapses exponentially; all classical layers upstream of the quantum encoding therefore receive vanishing gradients.

\paragraph{Empirical scope.} All experiments are conducted on statevector simulators with shot-noise-free measurement: we do not report results on physical hardware, under realistic noise channels, or on non-classification tasks (VQE, QAOA, generative modeling). The four-dataset benchmark of Table~\ref{tab:cross-dataset-eval} spans $n = 2$ to $12$, an early-onset rather than fully saturated barren plateau regime; only CIFAR-10 and the resource benchmarks of Figure~\ref{fig:membench} are extended to $n = 20$, which is where the barren plateau starts to become acute for global observables in generic HEAs~\cite{larocca2025barren}. At that scale RLQ-Grad matches, but does not exceed, dedicated mitigation methods such as layerwise learning (Table~\ref{tab:scale}), so its advantage over them is computational rather than in accuracy. Finally, the fixed per-iteration overhead of the full RLQ-Grad pipeline makes it slower than analytic differentiation on small circuits (below 6 to 8 qubits on CPU), and adjoint differentiation remains more memory-efficient at every scale we tested.

\paragraph{Transferability and agent cost.} The agents reported in Tables~\ref{tab:cross-dataset-eval} and \ref{tab:scale} are trained per configuration: a new agent is required for every choice of $(n, L)$ and dataset. The cost of training the agent alongside the QNN (rollouts, actor and critic updates, optimizer states) is included in Figure~\ref{fig:membench}, but it is not amortized across configurations. A learned-optimizer style meta-training regime that amortizes the agent across circuit scales and tasks (in the spirit of \cite{andrychowicz2016learning} and \cite{verdon2019learning}) is a natural next step but is outside the present scope.

\section{Conclusion}

We propose RLQ-Grad, a reinforcement-learning optimizer that replaces explicit differentiation of parameterized quantum circuits with a classical policy that directly proposes parameter updates. The key insight is that barren plateaus arise from the distribution of $\partial_\mu \ell$ induced by $U(\theta)$; by sampling updates from a policy $\pi_\phi$, RLQ-Grad bypasses this distribution entirely, avoiding the associated exponential variance collapse. Since $\pi_\phi$ is a finite-width MLP, its training follows standard classical concentration behavior independent of the number of qubits $n$, preventing analogous pathologies.

Trained with PPO and spectral normalization, RLQ-Grad achieves three main outcomes: (i) near-constant gradient variance (slope $\approx 0$) across $n=2$ to $12$, compared to decay rates of $-0.85$, $-0.30$, and $-0.14$ for backpropagation, parameter-shift, and adjoint differentiation; (ii) end-to-end time and memory that grow only linearly in the number of QNN parameters and are independent of the Hilbert space dimension, yielding speedups of up to $7876\times$ and under $2$ MB of memory at $n=20$ once the full agent training pipeline is accounted for; and (iii) accuracy gains of $2\%$ to $10\%$ over gradient-based baselines across four supervised benchmarks, together with parity with dedicated barren plateau mitigation methods on CIFAR-10 up to $20$ qubits, where evolutionary and gradient-free optimizers fail.

These results suggest that decoupling the optimization signal from quantum circuit geometry is a viable and complementary strategy to landscape-shaping approaches such as local costs, layerwise training, or structured initialization, matching their accuracy at a fraction of their simulation cost. However, RLQ-Grad does not address loss concentration, underparameterization, or gradient collapse in classical preprocessing layers. Future work includes meta-training across circuit scales, incorporating locality-aware rewards, extending to generative and variational tasks, and evaluating performance on noisy quantum hardware.

\begin{ack}
This research was financially sponsored in part by VinUniversity, Nhan Trong Luu was co-sponsored by both Can Tho University and VinUniversity.
\end{ack}

\bibliographystyle{ieeetr}
\bibliography{refs}

\appendix
\section{Universal Approximation Capability of HEA}\label{apd:uat}
 
The expressive power of the HEA circuit adopted in this work is grounded in the
Universal Approximation Theorem (UAT)~\cite{hornik1991approximation}, which we
restate for completeness.
 
\begin{theorem}[\cite{hornik1991approximation}]
Let $I_m = [0,1]^m$ and let $\varphi : \mathbb{R} \rightarrow \mathbb{R}$ be any
nonconstant, bounded, and continuous function. Then for every $f \in \mathcal{C}(I_m)$
and every $\varepsilon > 0$, there exist $N \in \mathbb{N}$, $\alpha_i, b_i \in \mathbb{R}$,
and $\vec{w}_i \in \mathbb{R}^m$ such that the single-hidden-layer network
\begin{equation}
    h(\vec{x}) = \sum_{i=1}^{N} \alpha_i\,\varphi\!\left(\vec{w}_i \cdot \vec{x} + b_i\right)
\end{equation}
satisfies $\sup_{\vec{x} \in I_m} |h(\vec{x}) - f(\vec{x})| < \varepsilon$.
\end{theorem}
 
PQCs admit an analogous universality result. This construction endows even single-qubit QNNs with universal approximation capability over compact domains~\cite{perez2020data, luu2026parameter}. The HEA used in this work inherits this property: with sufficient depth and data re-uploading, it can approximate arbitrary continuous classification functions to any prescribed precision.
 
\section{Quantum Barren Plateaus}\label{apd:barren_plateau}
 
Barren plateaus~\cite{larocca2025barren} refer to the exponential concentration of
PQC loss landscapes, wherein both output and gradient variances vanish exponentially in
the number of qubits $n$, rendering gradient-based optimization infeasible at scale.
 
\begin{definition}[\cite{larocca2025barren}]\label{def:barren_plateau}
Let $U(\theta)$ be a PQC acting on an $n$-qubit state $\rho$, and let
$f_{\theta}(\rho, O) = \mathrm{Tr}[O\,U(\theta)\rho U^{\dagger}(\theta)]$
for an observable $O$. The circuit is said to exhibit a \emph{barren plateau} over
base $b$ if
\begin{equation}
    \mathrm{Var}_{\theta}[f_{\theta}(\rho,O)]
    \;\;\text{or}\;\;
    \mathrm{Var}_{\theta}[\partial_{\mu}f_{\theta}(\rho,O)]
    \;\in\; \mathcal{O}\!\left(b^{-n}\right),
    \label{eq:prob_bp}
\end{equation}
where $b > 1$ depends on the choice of $\rho$ and $O$.
\end{definition}
 
The gradient variance bound in Equation \ref{eq:prob_bp} implies that the number of
circuit evaluations required to estimate $\partial_\mu f$ to fixed precision grows
exponentially in $n$, making training intractable on near-term hardware. This
concentration is mechanistically linked to the degree to which the parameter
distribution approximates a unitary $t$-design: deep HEA circuits approach a
$2$-design~\cite{park2024hardware}, which is sufficient to force
$\mathrm{Var}_\theta[\partial_\mu f] \in \mathcal{O}(2^{-n})$ for global
observables~\cite{larocca2025barren}. Mitigation strategies include the use of
local cost functions~\cite{cerezo2021cost}, structured
initialization~\cite{grant2019initialization}, layer-wise
training~\cite{skolik2021layerwise}, and hybrid architectures~\cite{luu2026parameter}.

\section{Proofs of main theorems}
\subsection{Proof of Theorem \ref{thm:bypass}}\label{proof:thm:bypass}
In standard gradient-based QNN training, the parameter update is:
\begin{equation}
    \theta_{t+1} = \theta_t - \eta\,\nabla_\theta\mathcal{L}(\theta_t),
\end{equation}
and the barren plateau renders $\nabla_\theta\mathcal{L}(\theta_t)\approx0$
exponentially in $n$, stalling optimisation\footnote{Noted in previous analysis of~\cite{mcclean2018barren}, where they characterised the barren plateau as $\mathrm{Var}_\theta[\partial_\mu\mathcal{L}] \in \mathcal{O}(2^{-n})$, arising from
approximate 2-design behaviour at large depth.}. In RLQ-Grad, the update is instead
\begin{equation}
    \theta_{t+1} = \theta_t - \eta\,g_t, \qquad g_t \sim \pi_\phi(\cdot\mid s_t).
\end{equation}
The quantity $g_t$ is the output of the policy network $\pi_\phi$, a purely classical
computation. Its distribution is determined by the agent parameters $\phi$ and the
state $s_t$, neither of which involves differentiating through $U(\theta)$.
Consequently, $g_t$ is statistically independent of $\nabla_\theta\mathcal{L}$, and
the barren plateau variance bound $\mathrm{Var}_\theta[\partial_\mu\mathcal{L}]
\in\mathcal{O}(b^{-n})$ places no constraint on $g_t$. The update therefore proceeds
regardless of circuit depth or qubit count.

The barren plateau condition Equation~\ref{eq:jaws_bp} is a statement about the distribution of $\partial_\mu\ell$ over the parameter ensemble induced by $U(\theta)$. In RLQ-Grad, the gradient applied to $\theta_t$ is $g_t = \pi_\phi(s_t)$, a function
of the RL agent's classical parameters $\phi$ and state $s_t$. Neither depends on
differentiating through the quantum circuit and consequently, the Wishart process
concentration result\footnote{Which is Equation~\ref{eq:jaws_bp}, originally derived from Corollary~4 of \cite{anschuetz2024unified}.} has no bearing on
$g_t$. The update proceeds regardless of the algebraic structure of $\mathcal{A}$. \cite{anschuetz2024unified} also show that
the precise scaling is governed by $\mathrm{dim}_\mathbb{R}(\mathrm{Aut}(\mathcal{A}_\alpha))$ and the entanglement structure of $\rho$ with respect to $\mathcal{A}$ (through $\mathrm{Tr}(\rho^\alpha)^2$), not solely by circuit depth.

Note that Theorem~\ref{thm:bypass} is also structural: it holds for any RL policy $\pi_\phi$ and any circuit $U(\theta)$, irrespective of whether the agent produces useful gradient estimates. It establishes only that the exponential vanishing pathology is absent from the update rule itself.

\subsection{Proof of Theorem \ref{thm:agent_grad}}\label{proof:thm:agent_grad}

The policy loss $\mathcal{L}_{\mathrm{RL}}(\phi)$ depends on $\phi$ only through the
classical computation $\pi_\phi(s_t)$. Differentiating via the chain rule gives
\begin{equation}
    \nabla_\phi \mathcal{L}_{\mathrm{RL}}
    = \sum_t \frac{\partial \mathcal{L}_{\mathrm{RL}}}{\partial g_t}
      \frac{\partial g_t}{\partial \phi},
\end{equation}
where $\partial g_t/\partial\phi$ is the Jacobian of the MLP output with respect to
its weights. This quantity is governed by classical backpropagation through a
feedforward network and satisfies standard gradient concentration bounds for
finite-width MLPs. In particular, its variance is polynomial in the network width $d$
and depth $L_{\mathrm{MLP}}$, not exponential in $n$. Since no quantum circuit
evaluation appears in this computation, no gradient variance decay stemming from quantum barren plateau can arise.\footnote{Still, note that Theorem~\ref{thm:agent_grad} does not assert that $\nabla_\phi\mathcal{L}_{\mathrm{RL}}$ is well-behaved in all respects. Classical MLPs can suffer from vanishing or exploding gradients when $L_{\mathrm{MLP}}$ is large, a problem addressed by standard techniques (batch normalisation, residual connections, careful initialisation). Critically, however, these pathologies are independent of $n$ and do not worsen as the quantum system scales.}

\section{Miscellaneous and settings}\label{apd:misc}

\subsection{Preliminary studies and agent training settings}

In our main experiments, we train the agent using the Proximal Policy Optimization (PPO)~\cite{schulman2017proximalpolicyoptimizationalgorithms} algorithm with the default multilayer perceptron policy as implemented in Stable-Baselines3 (SB3)~\cite{stable-baselines3}. Justifications for this selection are made based on the results of our preliminary study across different model types, as noted in Table~\ref{tab:arch_bench}.

\begin{table*}[t]
    \caption{Preliminary studies on different agent type with RLQ-Grad. $l_{\mathrm{rollout}}$ is correspondingly 400 for BC, 46800 for MNIST and FMNIST, 39000 for CIFAR10.}
    \centering
    \resizebox{\textwidth}{!}{
        \begin{tabular}{@{}lllll*{4}{c}@{}}
        \toprule
        \multirow{3}{*}{\bf Agents} & \multirow{3}{*}{\bf Buffer type}& \multirow{3}{*}{\bf Buffer size} & \multirow{3}{*}{\bf Params} & \multirow{3}{*}{\bf MACs} & \multicolumn{4}{c}{\bf Top 1 validation accuracy per dataset}\\ 
        \cmidrule{6-9}
        & & & & & BC & MNIST & F-MNIST & CIFAR10\\
        \midrule
        SAC & Replay & 1000000 & 17,038 & 16,640 & $95.61\pm1.519$ & $65.53\pm0.521$ & $72.27\pm0.437$ & $28.64\pm0.516$\\
        TD3 & Replay & 1000000 & 16,648 & 16,256 & $96.71\pm0.380$ & $68.55\pm0.989$ & $73.52\pm0.367$ & $28.72\pm0.287$\\
        DDPG & Replay & 1000000 & 11,079 & 10,816 & $85.53\pm6.837$ & $55.23\pm0.372$ & $65.28\pm0.098$ & $23.14\pm0.457$\\
        TRPO & Rollout & $l_{\mathrm{rollout}}$ & 10,695 & 10,432 & $93.20\pm1.899$  & $60.27\pm0.298$ & $70.65\pm0.428$ & $26.46\pm0.623$\\
        PPO & Rollout & $l_{\mathrm{rollout}}$ & 5,575 & 5,440 & $90.05\pm4.760$ & $59.52\pm0.212$ & $68.52\pm0.534$ & $25.12\pm0.961$\\
        \bottomrule
        \end{tabular}
    }
    \label{tab:arch_bench}
    \vspace{-0.7\intextsep}
\end{table*}

Table~\ref{tab:arch_bench} shows a consistent performance hierarchy across all datasets, with off-policy methods (Soft Actor-Critic (SAC)~\cite{haarnoja2018soft}, Twin Delayed Deep Deterministic (TD3)~\cite{fujimoto2018addressing}, Deep Deterministic Policy Gradient (DDPG)~\cite{lillicrap2020continuous}) outperforming on-policy counterparts (Trust Region Policy Optimization (TRPO)~\cite{schulman2015trust} and PPO). In particular, TD3 achieves the best overall accuracy on all four datasets, followed closely by SAC, while DDPG exhibits the weakest performance among replay-based methods. On-policy methods (TRPO and PPO) consistently lag behind, indicating a performance gap between replay-buffer-based and rollout-based training strategies under the same architectural constraints.

A clear distinction emerges between replay and rollout buffers. Replay-based agents (SAC, TD3, DDPG) significantly outperform rollout-based agents (TRPO, PPO) across all datasets. This trend is especially pronounced on more complex datasets such as CIFAR10, where TD3 and SAC achieve approximately $28.7\%$ accuracy, compared to $26.46\%$ for TRPO and $25.12\%$ for PPO.

Despite relatively small differences in parameter count and MACs, performance varies substantially. TD3 and SAC have comparable computational costs ($\sim$16k parameters and MACs), yet TD3 consistently achieves slightly better accuracy with lower variance. DDPG, while having fewer parameters ($\sim$11k), suffers from significantly degraded performance and higher variance.

All methods achieve high accuracy on the simplest dataset (BC), with TD3 reaching $96.71\%$, while performance drops significantly as dataset complexity increases. The relative ordering of methods, however, remains consistent across datasets. Variance across runs further differentiates the methods. TD3 demonstrates both high accuracy and low variance, indicating stable training dynamics. In contrast, DDPG and PPO exhibit larger standard deviations (e.g., $6.837$ and $4.760$ on BC).

On the other hand, PPO achieves competitive performance relative to TRPO despite having nearly half the parameters and MACs (5,575 vs. 10,695), making them an interesting candidate to further improve on. Furthermore, while achieving higher accuracy in overall, agents that used replay buffer (such as SAC, TD3 and DDPG) are significantly more memory consuming on CPU/GPU in comparison with those that use rollout buffers, eventually failing the promise of low resource consumption that we made initially.

All benchmarked policy and value networks are jointly optimized using stochastic gradient descent with a fixed learning rate of $5e-4$. The minibatch size is statically set to be equal with $\mathrm{total \textunderscore supervised\textunderscore training \textunderscore epoch}$\footnote{Which is around 200, meaning that the rollout buffer is shuffled and split into batches of 200 samples for agent gradient updates with SGD.}. To preserve memory efficient and avoid asynchronous gradient propose/model update problem, we only use 1 environment to train all the RL agents. At each training iteration, all agent that use rollout buffer collects a rollout of length (rollout buffer size) equal to the total training steps, calculated as:
\begin{equation}
    l_{\mathrm{rollout}} = \mathrm{int}\left(\frac{\mathrm{total \textunderscore training \textunderscore sample}}{\mathrm{supervised \textunderscore training \textunderscore batch \textunderscore size}}\right) \cdot \mathrm{total \textunderscore supervised\textunderscore training \textunderscore epoch},
\end{equation}
after which the collected transitions are used to perform multiple epochs of minibatch updates. This rollout configuration is motivated by the observation that larger values of $l_{\mathrm{rollout}}$ lead to more stable training dynamics and ultimately improved convergence, even over prolonged training durations. Benchmarking results that support this claim using TRPO and PPO on MNIST can be found in Table~\ref{tab:rollout_ablation}.

\begin{table*}[h]
    \caption{Accuracy comparison of PPO and TRPO across $l_{\mathrm{rollout}}$ length on MNIST.}
    \centering
    \resizebox{0.8\width}{!}{
        \begin{tabular}{@{}l*{5}{c}@{}}
        \toprule
        \multirow{3}{*}{\bf Agents} & \multicolumn{5}{c}{\bf Top 1 validation accuracy per dataset}\\ 
        \cmidrule{2-6}
        & $0.2 \cdot l_{\mathrm{rollout}}$ & $0.4 \cdot l_{\mathrm{rollout}}$ & $0.6 \cdot l_{\mathrm{rollout}}$ & $0.8 \cdot l_{\mathrm{rollout}}$ & $l_{\mathrm{rollout}}$\\
        \midrule
        TRPO  & $57.13\pm0.199$ & $58.35\pm0.467$ & $58.73\pm0.512$ & $59.35\pm0.408$ & $60.27\pm0.298$\\
        PPO & $57.22\pm0.204$ & $57.62\pm0.259$ & $58.79\pm0.315$ & $58.94\pm0.342$ & $59.52\pm0.212$\\
        \bottomrule
        \end{tabular}
    }
    \label{tab:rollout_ablation}
\end{table*}

We also fixed agent initialization random seed to ensure reproducibility across runs. Unless otherwise specified, all other agent hyperparameters (e.g., clipping range, number of epochs, GAE parameters, and entropy coefficient) are kept at their default values provided by SB3.

\subsection{Stablization tricks with spectral normalization}\label{sec:sn_norm_vf}

Training RL agents with function approximation is notoriously unstable, particularly in the bootstrapped regime of temporal-difference (TD) learning where approximation errors compound through iterated Bellman updates~\cite{sutton1998reinforcement}. Spectral normalization (SN) has emerged as a principled remedy for this instability, with demonstrated effectiveness in stabilizing Q-networks~\cite{gogianu2021spectral} and actor-critic architectures~\cite{bjorck2021towards, luu2026robust}. To avoid this problem, we can adopt SN as a method for stabilizing the training process of our algorithm. Consider a feed-forward network $\mathcal{F}_\theta$ with $L$ layers:
\begin{equation}
    z_i = W_i\, a_{i-1} + b_i, \qquad a_i = \sigma(z_i), \qquad i = 1, \dots, L,
\end{equation}
where $\sigma(\cdot)$ is the ReLU activation for all hidden layers ($i < L$). For a designated subset of layers $\mathcal{S} \subseteq \{1, \dots, L\}$, SN replaces each weight matrix with its spectrally normalized counterpart:
\begin{equation}\label{eq:sn_weight}
    \hat{W}_i = \frac{W_i}{\rho_i}, \qquad \rho_i = \max\bigl(\rho(W_i),\; k\bigr),
\end{equation}
where $\rho(W_i) = \sigma_{\max}(W_i)$ denotes the largest singular value (spectral norm) of $W_i$, $k > 0$ is a small constant preventing degenerate scaling, and a stop-gradient is applied to $\rho_i$ during backpropagation so that normalization acts as a constraint rather than a differentiable transformation.
 
By capping the spectral norm of each layer, SN enforces a Lipschitz bound on the end-to-end mapping. Specifically, for a network with ReLU activations, the Lipschitz constant satisfies:
\begin{equation}\label{eq:lipschitz_bound}
    \mathrm{Lip}(\mathcal{F}_\theta) \;\leq\; \prod_{i \in \mathcal{S}} \rho_i^{-1} \cdot \prod_{i=1}^{L} \|W_i\|_2 \;\leq\; \prod_{i \notin \mathcal{S}} \|W_i\|_2,
\end{equation}
which directly limits how rapidly the network output can change with respect to its input. In the RL context, this Lipschitz constraint serves two critical purposes: (i)~it bounds the propagation of approximation error through bootstrapped TD targets, preventing the runaway divergence observed in unconstrained networks~\cite{bjorck2021towards}; and (ii)~it promotes smoother value landscapes, improving generalization across nearby states.
 
Beyond its regularization effect, SN induces a data-dependent modulation of gradient magnitudes that fundamentally alters the optimization dynamics. Under a bias-scaling reformulation~\cite{gogianu2021spectral}, the pre-activations of the normalized network satisfy:
\begin{equation}
    \hat{z}_i = \rho_{1:i}^{-1}\, z_i, \qquad \text{where} \quad \rho_{1:i} = \prod_{\substack{j \leq i,\; j \in \mathcal{S}}} \rho_j.
\end{equation}
The loss computed on the normalized network and the gradient with respect to the unnormalized parameters is therefore:
\begin{equation}
    \hat{\mathcal{L}} = \mathcal{L}\!\bigl(\rho_{1:L}^{-1}\, z_L\bigr),\quad 
    \frac{\partial \hat{\mathcal{L}}}{\partial W_i} = \rho_{1:L}^{-1}\, J_i\, \hat{\delta}_L\, a_{i-1}^\top,
    \label{eq:grad_rescale}
\end{equation}
where $J_i$ is the Jacobian of all layers subsequent to layer $i$ and $\hat{\delta}_L = \partial \hat{\mathcal{L}} / \partial \hat{z}_L$. Equation~\ref{eq:grad_rescale} reveals that SN imposes a \emph{global} rescaling of all gradients by $\rho_{1:L}^{-1}$, the inverse product of spectral norms across normalized layers. As network weights grow during training (a natural consequence of fitting increasingly complex value landscapes), $\rho_{1:L}$ increases monotonically, progressively attenuating gradient magnitudes. This gradient rescaling has a particularly consequential interaction with adaptive optimizers. Recall the Adam update rule~\cite{kingma2014adam}:
\begin{equation}
    \Delta\theta_t = \eta \cdot \frac{\hat{m}_t}{\sqrt{\hat{v}_t} + \varepsilon},
\end{equation}
where $\hat{m}_t$ and $\hat{v}_t$ are bias-corrected estimates of the first and second gradient moments, respectively. When gradients are globally rescaled as $g_t' = \rho_t^{-1}\, g_t$ (with $\rho_t \coloneqq \rho_{1:L}$ at iteration $t$), the moment estimates transform as:
\begin{equation}
    \hat{m}_t' \approx \rho_t^{-1}\, \hat{m}_t, \qquad \hat{v}_t' \approx \rho_t^{-2}\, \hat{v}_t.
\end{equation}
Substituting into the Adam update yields:
\begin{equation}\label{eq:adam_modified}
    \Delta\theta_t' \approx \eta \cdot \frac{\hat{m}_t}{\sqrt{\hat{v}_t} + \rho_t\, \varepsilon}.
\end{equation}
The key observation is that SN effectively multiplies the damping constant $\varepsilon$ by $\rho_t$, which grows monotonically over training. This induces a natural three-phase optimization trajectory:
\begin{enumerate}[leftmargin=*]
    \item \textbf{Early training} ($\rho_t \approx 1$): The modified update Equation~\ref{eq:adam_modified} reduces to standard Adam. The optimizer retains its full adaptive behavior, enabling rapid initial learning and efficient exploration of the loss landscape.
    \item \textbf{Mid training} ($\rho_t$ growing): The inflated denominator uniformly reduces step sizes across all parameters, acting as a stabilizer that dampens sensitivity to noisy gradient estimates, which is particularly important in RL, where minibatch composition and advantage estimation introduce high variance.
    \item \textbf{Late training} ($\rho_t\, \varepsilon \gg \sqrt{\hat{v}_t}$): The update asymptotically approaches
    \begin{equation}
        \Delta\theta_t' \;\approx\; \frac{\eta}{\rho_t\, \varepsilon}\, \hat{m}_t,
    \end{equation}
    which is equivalent to SGD with momentum at a decaying effective learning rate $\eta_{\mathrm{eff}} = \eta / (\rho_t\, \varepsilon)$. This automatic transition from adaptive to momentum-based updates is highly desirable: adaptive methods excel at navigating the complex early-training landscape but suffer from poor curvature estimates in the non-stationary TD regime, whereas SGD with momentum provides more stable convergence near optima~\cite{gogianu2021spectral}.
\end{enumerate}
 
This optimization-centric interpretation is empirically validated by~\cite{gogianu2021spectral}, who show that explicit schedulers mimicking either the gradient rescaling (DIVGRAD) or the $\varepsilon$-scaling (MULEPS) of SN recover comparable performance gains, even in regimes where full SN fails to train.
 
We apply SN selectively to the value and policy networks following the power iteration scheme described in Algorithm~\ref{alg:sn_rl_training_general}. Each normalized layer maintains a pair of auxiliary vectors $(\mathbf{u}, \mathbf{v})$ that approximate the leading left and right singular vectors of $W_i$ via a single power iteration step per training update (with a negligible computational overhead).
 
\begin{algorithm}[t]
    \caption{RL Training with Spectral Normalization}
    \label{alg:sn_rl_training_general}
    \begin{algorithmic}[1]
    \REQUIRE Dataset $\mathcal{B}$ (replay buffer or on-policy batch), trainable networks $\mathcal{F}_\theta$, normalized layer set $\mathcal{S}$
    \FOR{each training iteration}
        \STATE Sample batch $\mathcal{T} \sim \mathcal{B}$
        \FOR{each layer $i$ in $\mathcal{F}_\theta$}
            \IF{$i \in \mathcal{S}$}
                \IF{$\mathbf{u}^{(i)}, \mathbf{v}^{(i)}$ not yet initialized}
                    \STATE Initialize $\mathbf{u}^{(i)}, \mathbf{v}^{(i)} \sim \mathcal{N}(0, I)$
                \ENDIF
                \STATE $\tilde{\mathbf{v}} \leftarrow W_i^\top\, \mathbf{u}^{(i)}$; \quad $\mathbf{v}^{(i)} \leftarrow \tilde{\mathbf{v}} / \|\tilde{\mathbf{v}}\|$
                \STATE $\tilde{\mathbf{u}} \leftarrow W_i\, \mathbf{v}^{(i)}$; \quad $\rho_i \leftarrow \|\tilde{\mathbf{u}}\|$; \quad $\mathbf{u}^{(i)} \leftarrow \tilde{\mathbf{u}} / \rho_i$
                \STATE $\bar{W}_i \leftarrow W_i \,/\, \max(1,\; \rho_i)$ \hfill \COMMENT{Clamp spectral norm to $\leq 1$}
            \ENDIF
        \ENDFOR
        \STATE Compute RL loss $\mathcal{L}_{\mathrm{RL}}(\mathcal{T};\, \mathcal{F}_\theta)$ using normalized weights $\bar{W}$
        \STATE Update $\theta$ via gradient descent on $\mathcal{L}_{\mathrm{RL}}$
    \ENDFOR
    \end{algorithmic}
\end{algorithm}
 
\begin{figure*}[t!]
    \begin{subfigure}[t]{.325\textwidth}
      \centering
        \includegraphics[width=\textwidth]{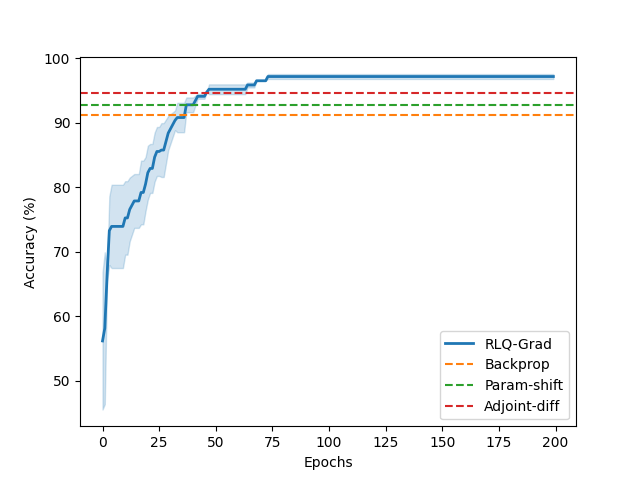}
      \caption{PolicySN[-1] + ValueSN[-1].}
    \end{subfigure}
    \hfill
    \begin{subfigure}[t]{.325\textwidth}
      \centering
        \includegraphics[width=\textwidth]{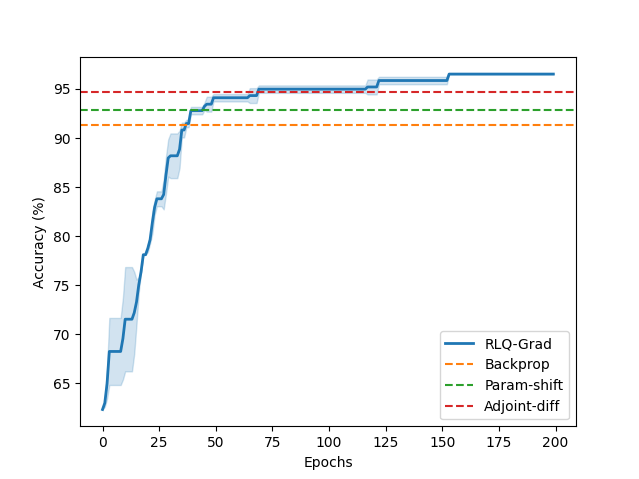}
      \caption{ValueSN[-1].}
    \end{subfigure}
    \hfill
    \begin{subfigure}[t]{.325\textwidth}
      \centering
        \includegraphics[width=\textwidth]{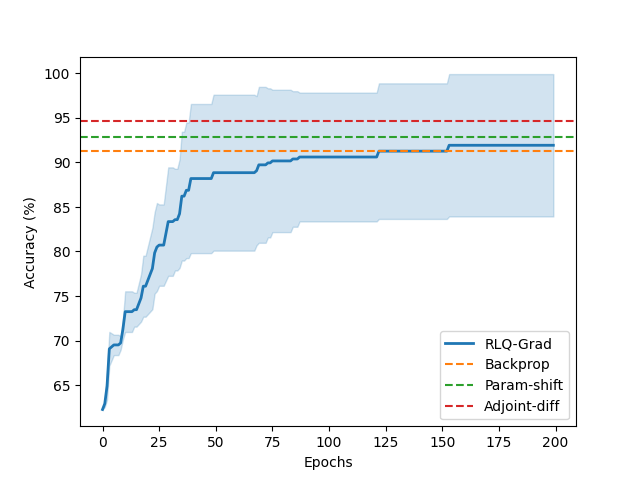}
      \caption{No SN.}
    \end{subfigure}
    \caption{Accuracy convergence of RLQ-Grad (with PPO) variants (with (a) both last layer of policy and value function SN, (b) only last layer of value function SN~\cite{gogianu2021spectral} and (c) no SN applied) on BC dataset (denoted blue with fluctuation region) when compare with maximum achievable accuracy of other gradient calculating methods on 2-qubits HEA.}
    \label{fig:sn_ablation}
\end{figure*}
 
To understand the contribution of SN on different network components, we evaluate three configurations on the BC dataset for the 2-qubit HEA problem (Figure~\ref{fig:sn_ablation}): SN applied to the last layer of both the policy and value networks (PolicySN[-1] + ValueSN[-1]), SN applied to only the last layer of the value network (ValueSN[-1]~\cite{gogianu2021spectral}), and no SN applied.
 
The most striking difference across configurations is in convergence reliability. Configuration~(a), with SN on both networks, achieves the fastest and most stable convergence, reaching near-optimal accuracy ($>93\%$) within approximately 50 training epochs with minimal variance across runs. Configuration~(b), with SN on the value network alone, also converges to comparable final accuracy, albeit with a slightly delayed onset, consistent with the value network being the primary source of bootstrap-induced instability. In contrast, configuration~(c) without SN exhibits dramatically higher variance (visible as the wide blue fluctuation band) and slower, less reliable convergence, with some runs failing to reach the performance plateau within the allotted training budget.
 
The comparison between configurations~(b) and~(c) isolates the stabilizing effect of SN on the critic. The value network directly participates in computing TD targets, making it the primary channel through which approximation errors are amplified. By constraining the Lipschitz constant of the critic, SN bounds the magnitude of these compounding errors (Equation~\ref{eq:lipschitz_bound}), preventing the catastrophic overestimation that often destabilizes training. The near-identical final performance of configurations~(a) and~(b) suggests that value network normalization accounts for the majority of the stability gains.
 
While policy network SN is not strictly necessary for convergence, configuration~(a) demonstrates a measurable advantage in convergence speed over~(b). We attribute this to the implicit learning rate scheduling effect that normalizing the policy network introduces an additional spectral radius factor into the gradient rescaling (Equation~\ref{eq:grad_rescale}), providing finer-grained control over the policy update magnitude.
 
Across all three configurations, the converged accuracy of our RLQ-Grad method matches or exceeds the maximum achievable accuracy of alternative gradient calculation methods (Backprop, Param-shift; shown as dashed horizontal lines in Figure~\ref{fig:sn_ablation}). Notably, even the unstabilized variant~(c) can achieve competitive final accuracy in successful runs, but its high variance makes it unreliable in practice.

\begin{table*}[h]
    \caption{Accuracy comparison between ValueSN[-1] and non-SN PPO across datasets.}
    \centering
    \resizebox{0.8\textwidth}{!}{
        \begin{tabular}{@{}lll*{4}{c}@{}}
        \toprule
        \multirow{3}{*}{\bf Agents} & \multirow{3}{*}{\bf Params} & \multirow{3}{*}{\bf MACs} & \multicolumn{4}{c}{\bf Top 1 validation accuracy per dataset}\\ 
        \cmidrule{4-7}
        & & & BC & MNIST & F-MNIST & CIFAR10\\
        \midrule
        PPO & 5,575 & 5,440 & $90.05\pm4.760$ & $59.52\pm0.212$ & $68.52\pm0.534$ & $25.12\pm0.961$\\
        \bf PPO (SN) & 11847 & 10,432 & $95.52\pm0.080$ & $61.36\pm2.216$ & $73.64\pm0.415$ & $27.17\pm0.782$\\
        \bottomrule
        \end{tabular}
    }
    \label{tab:sn_ablation}
\end{table*}

Table~\ref{tab:sn_ablation} demonstrates that incorporating SN into PPO yields consistent improvements across all datasets. PPO (SN) outperforms the baseline PPO in every case, with particularly notable gains on BC ($+5.47\%$) and F-MNIST ($+5.12\%$). Even on more challenging datasets such as CIFAR10, SN provides a measurable improvement ($+2.05\%$). The magnitude of improvement varies across datasets. Larger gains are observed on intermediate-complexity datasets (e.g., F-MNIST), while improvements are more modest on MNIST and CIFAR10.

The improvements in accuracy come at the cost of increased model complexity. PPO (SN) approximately doubles both the parameter count (11,847 vs.\ 5,575) and MACs (10,432 vs.\ 5,440). Despite this, the performance gains are disproportionately large relative to the increase in computational cost, particularly on BC and F-MNIST.

\subsection{Input dimensionality reduction}\label{apd:pca}

We employ Principal Component Analysis (PCA) to project the high-dimensional input into a lower-dimensional subspace while preserving maximal global variance. To determine the optimal number of features $k$, we evaluate the Cumulative Explained Variance Ratio (CEVR), defined as $\text{CEVR}(k) = \frac{\sum_{j=1}^{k} \lambda_j}{\sum_{j=1}^{d} \lambda_j}$ where $\lambda_j$ represents the eigenvalue associated with the $j$-th principal component, $d$ is the input dimension and captures approximately 95\% of the total variance.

\subsection{Accuracy convergence curves}\label{apd:acc_curves}

Figure~\ref{fig:acc_main} shows the validation accuracy convergence of RLQ-Grad on the 2-qubit HEA for the four datasets of Table~\ref{tab:cross-dataset-eval}, compared with the maximum accuracy achieved by the other gradient computation methods.

\begin{figure*}[h]
    \begin{subfigure}[t]{.245\textwidth}
      \centering
        \includegraphics[width=\textwidth]{ppo_bc_2b_vfsn.png}
      \caption{BC.}
    \end{subfigure}
    \hfill
    \begin{subfigure}[t]{.245\textwidth}
      \centering
        \includegraphics[width=\textwidth]{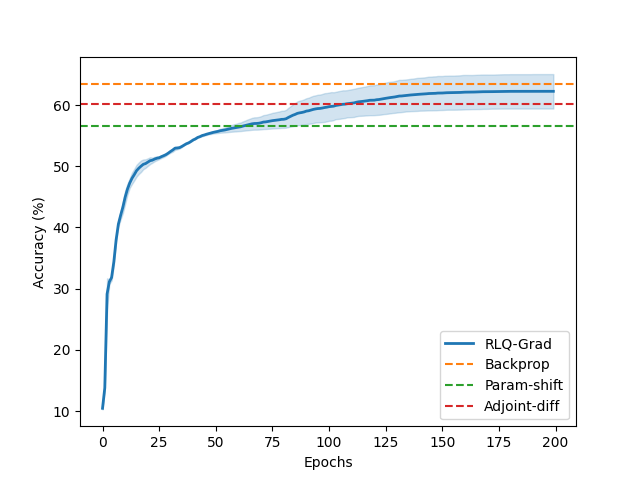}
      \caption{MNIST dataset.}
    \end{subfigure}
    \hfill
    \begin{subfigure}[t]{.245\textwidth}
      \centering
        \includegraphics[width=\textwidth]{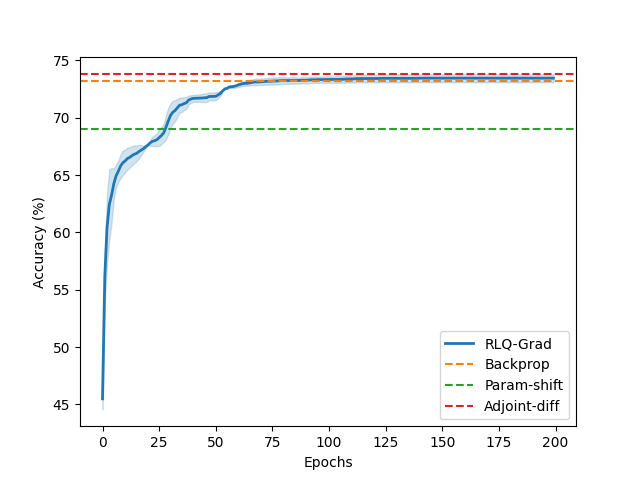}
      \caption{F-MNIST dataset.}
    \end{subfigure}
    \hfill
    \begin{subfigure}[t]{.245\textwidth}
      \centering
        \includegraphics[width=\textwidth]{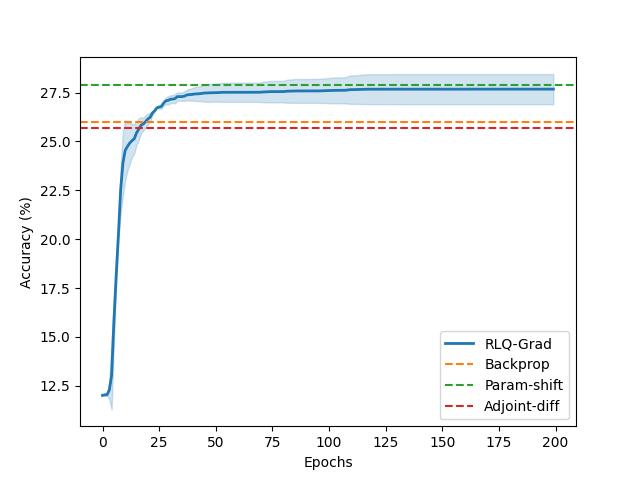}
      \caption{CIFAR-10 dataset.}
    \end{subfigure}
    \caption{Accuracy convergence of RLQ-Grad (denoted blue with fluctuation region) when compare with maximum achievable accuracy of other gradient calculating methods on 2-qubits HEA.}
    \label{fig:acc_main}
    \vspace{-0.8\intextsep}
\end{figure*}

\subsection{Hardware specifications}

Since adjoint differentiation and parameter-shift rule currently lack efficient GPU implementations comparable to backpropagation and our proposed method, the resource allocation benchmark in Figure~\ref{fig:membench} was conducted on a CPU (AMD Ryzen 5 5600G) to assure fairness, averaging 12 timed trials after 4 warm-up iterations. The GPU comparison between RLQ-Grad and backpropagation (Table~\ref{tab:gpu_time}) was run on an NVIDIA RTX 3090 and averaged over 15 trials. In contrast, for accuracy evaluations where minimizing training time is critical, experiments for our method and backpropagation were performed on an NVIDIA RTX 3090 (24GB) GPU.

\subsection{Full resource benchmark values}\label{apd:resource_tables}

Tables~\ref{tab:cpu_mem}, \ref{tab:cpu_time}, and \ref{tab:gpu_time} list the exact values behind Figure~\ref{fig:membench} and the GPU comparison. Rather than measuring only the actor's forward pass, all RLQ-Grad entries account for the complete training pipeline (PPO rollout, policy and value updates, gradient buffers, and Adam states).

\begin{table*}[h]
    \caption{End-to-end CPU RAM usage of gradient computation methods ($n=2$ to $20$ qubits).}
    \centering
    \resizebox{\textwidth}{!}{
        \begin{tabular}{@{}l*{10}{r}@{}}
        \toprule
        Method & 2q1d & 4q2d & 6q3d & 8q4d & 10q5d & 12q6d & 14q7d & 16q8d & 18q9d & 20q10d\\
        \midrule
        Backpropagation & $<0.1$ MB & 187.15 KB & 416.55 KB & 619.74 KB & 1.96 MB & 10.10 MB & 63.30 MB & 282.70 MB & 1227 MB & 6174 MB\\
        Parameter-shift & $<0.1$ MB & 282.89 KB & 1.06 MB & 6.38 MB & 7.32 MB & 9.98 MB & 18.59 MB & 22.37 MB & 24.23 MB & 22.55 MB\\
        Adjoint-differentiation & $<0.1$ MB & $<0.1$ MB & $<0.1$ MB & $<0.1$ MB & $<0.1$ MB & $<0.1$ MB & $<0.1$ MB & $<0.1$ MB & $<0.1$ MB & $<0.1$ MB\\
        \bf RLQ-Grad (full) & \bf 161.22 KB & \bf 215.64 KB & \bf 306.35 KB & \bf 433.33 KB & \bf 596.60 KB & \bf 796.14 KB & \bf 1.02 MB & \bf 1.30 MB & \bf 1.61 MB & \bf 1.95 MB\\
        \bottomrule
        \end{tabular}
    }
    \label{tab:cpu_mem}
\end{table*}

\begin{table*}[h]
    \caption{Per-iteration CPU wall-clock time (in seconds) of gradient computation methods.}
    \centering
    \resizebox{\textwidth}{!}{
        \begin{tabular}{@{}l*{10}{r}@{}}
        \toprule
        Method & 2q1d & 4q2d & 6q3d & 8q4d & 10q5d & 12q6d & 14q7d & 16q8d & 18q9d & 20q10d\\
        \midrule
        Backpropagation & 0.0042 & 0.0123 & 0.0260 & 0.1295 & 0.1033 & 0.1662 & 0.2719 & 0.5112 & 2.3262 & 219.41\\
        Parameter-shift & 0.0067 & 0.0425 & 0.1333 & 0.3921 & 0.9860 & 2.3931 & 5.2019 & 16.2818 & 66.3288 & 693.96\\
        Adjoint-differentiation & 0.0027 & 0.0055 & 0.0109 & 0.0161 & 0.0266 & 0.0569 & 0.0836 & 0.6666 & 10.0897 & 59.3241\\
        \bf RLQ-Grad (full) & \bf 0.0455 & \bf 0.0545 & \bf 0.0641 & \bf 0.0565 & \bf 0.0815 & \bf 0.0858 & \bf 0.0842 & \bf 0.0862 & \bf 0.0875 & \bf 0.0881\\
        \bottomrule
        \end{tabular}
    }
    \label{tab:cpu_time}
\end{table*}

\begin{table*}[h]
    \caption{Per-iteration GPU wall-clock time (in milliseconds, NVIDIA RTX 3090, 15 trials). Parameter-shift and adjoint differentiation are omitted as they lack efficient GPU kernels in current QNN frameworks.}
    \centering
    \resizebox{\textwidth}{!}{
        \begin{tabular}{@{}l*{10}{r}@{}}
        \toprule
        Method & 2q1d & 4q2d & 6q3d & 8q4d & 10q5d & 12q6d & 14q7d & 16q8d & 18q9d & 20q10d\\
        \midrule
        Backpropagation & 23.37 & 65.54 & 126.96 & 218.95 & 322.98 & 451.28 & 605.74 & 807.99 & 1026.45 & 1457.06\\
        \bf RLQ-Grad (full) & \bf 22.51 & \bf 56.60 & \bf 111.49 & \bf 184.76 & \bf 276.88 & \bf 390.89 & \bf 523.60 & \bf 785.93 & \bf 957.31 & \bf 1190.83\\
        \bottomrule
        \end{tabular}
    }
    \label{tab:gpu_time}
\end{table*}

\section{Derivation process of quantum gradient calculation methods}\label{apd:complexity}

Assuming we have a QNN with $n$ qubits and $L$ parameterized gates (or layers), let $P$ denote the number of trainable parameters, typically $P = \mathcal{O}(L)$. A quantum state on $n$ qubits is represented as a complex vector $|\psi\rangle \in \mathbb{C}^{2^n}$, so storing a single state requires a memory complexity of $\mathcal{O}_{\mathrm{memory}}(2^n)$.

\subsection{Forward pass complexity}

Every gate acts on the full statevector, leading to a computational cost of $\mathcal{O}(2^n)$ per gate application. Therefore, the total cost of the forward pass is $\mathcal{O}(L \cdot 2^n)$, and this factor remain unchanged since benchmarked optimization algorithms only affect gradient computations.

\subsection{Backward pass complexity}

\subsubsection{Backpropagation}

Let the circuit be composed of unitary operators $\{U_k\}_{k=1}^L$. The forward evolution is $|\psi_k\rangle = U_k |\psi_{k-1}\rangle$ with $k = 1, \dots, L$. Define the adjoint state $|\lambda_L\rangle = O |\psi_L\rangle$ where $O$ is the observable; the backward recursion is:
\begin{equation}
    |\lambda_{k-1}\rangle = U_k^\dagger |\lambda_k\rangle.
\end{equation}
For a parameterized gate $U_k(\theta_k) = e^{-i \theta_k G_k}$, we have
$\tfrac{dU_k}{d\theta_k} = -iG_k U_k$, so differentiating
$\langle O \rangle = \langle\psi_L|O|\psi_L\rangle$ and collecting both the bra and ket
contributions yields:
\begin{equation}
\frac{\partial \langle O \rangle}{\partial \theta_k}
= 2\,\mathrm{Im} \left(
\langle \lambda_k | G_k | \psi_k \rangle
\right).
\end{equation}
With $L$ evolution steps we obtain a time complexity of $\mathcal{O}_{\mathrm{time}}(L \cdot 2^n)$. For memory complexity, since all intermediate states $\{|\psi_k\rangle\}_{k=0}^L$ must be stored, we have $\mathcal{O}_{\mathrm{memory}}(L \cdot 2^n)$.

\subsubsection{Parameter-Shift Rule}

When using parameter-shift~\cite{mitarai2018quantum}, for a parameter $\theta_i$ the gradient is computed as:
\begin{equation}
\frac{\partial \langle O \rangle}{\partial \theta_i}
= \frac{1}{2} \left[
f\left(\theta_i + \frac{\pi}{2}\right)
- f\left(\theta_i - \frac{\pi}{2}\right)
\right],
\end{equation}
where each evaluation requires a full forward pass. Thus, computing the gradient for all parameters requires $\mathcal{O}_{\mathrm{time}}(P \cdot L \cdot 2^n)$. Since only one statevector is needed at a time per parameter update, the memory complexity is theoretically $\mathcal{O}_{\mathrm{memory}}(2^n)$. Still, in practice many framework such as Pennylane~\cite{bergholm2018pennylane} accumulate all parameters gradient and update the model as a whole, leading to a practical memory complexity of $\mathcal{O}_{\mathrm{memory}}(P_{\mathrm{QNN}} \cdot 2^n)$.

\subsubsection{Adjoint differentiation}

The reverse-mode gradient computation derived in \cite{jones2020efficient} (or for short "adjoint differentiation" as noted in some materials \cite{lee2021adjoint}) is mathematically equivalent to traditional backpropagation, but exploits the reversible structure of quantum circuits to achieve improved memory efficiency.

In classical neural networks, backpropagation computes gradients via reverse-mode accumulation:
\begin{equation}
\boldsymbol{\delta}_k
= \left( \frac{\partial \mathbf{h}_{k+1}}{\partial \mathbf{h}_k} \right)^\top
\boldsymbol{\delta}_{k+1},
\end{equation}
requiring storage of all intermediate activations $\{\mathbf{h}_k\}_{k=1}^L$, leading to memory scaling proportionally with $L$. In quantum neural networks, the analogous reverse propagation is:
\begin{equation}
|\lambda_{k-1}\rangle = U_k^\dagger |\lambda_k\rangle,
\end{equation}
which is directly analogous to multiplication by the transposed Jacobian in classical backpropagation. The gradient of the expectation value $\langle E \rangle$ with respect to parameter $\theta_i$ is:
\begin{equation}
\frac{\partial \langle E \rangle}{\partial \theta_i}
= 2 \,\mathrm{Re} \left(
\langle \lambda_i |
\frac{dU_i}{d\theta_i}
| \phi_i \rangle
\right),
\end{equation}
where the backward state $|\lambda_i\rangle$ and the reconstructed forward state $|\phi_i\rangle$
satisfy the recursions (traversed from $k=L$ down to $k=1$):
\begin{equation}
    |\phi_i\rangle = U_i^\dagger |\phi_{i+1}\rangle, \quad |\lambda_i\rangle = U_{i+1}^\dagger |\lambda_{i+1}\rangle.
\end{equation}
Here $|\phi_{i+1}\rangle$ is the reconstructed forward state at layer $i+1$,
obtained by applying the inverse gate $U_{i+1}^\dagger$ to the current forward state
rather than storing it from the forward pass.

Because quantum circuits are unitary and therefore reversible, the intermediate forward
states $\{|\psi_k\rangle\}$ need not be stored; they are recomputed on-the-fly during
the backward sweep. Consequently, all $P$ gradients are obtained in a single reverse
pass at a cost of $\mathcal{O}_{\mathrm{time}}(P_{QNN} \cdot 2^n)$, matching backpropagation's
time complexity while reducing the memory required beyond the forward statevector to $\mathcal{O}_{\mathrm{memory}}(1)$, since only a fixed number of working buffers (the current forward and backward states) are kept regardless of $L$ and $P_{\mathrm{QNN}}$.

\subsection{RLQ-Grad}\label{apd:rlq_complexity}

In the case of our RLQ-Grad, gradient evaluation reduces to a forward pass through the MLP actor. For an $L_{\mathrm{MLP}}$-layer network of hidden width $d$, whose input $s_t$ and output $g_t$ both have dimension $\mathcal{O}(P_{\mathrm{QNN}})$, the input layer costs $\mathcal{O}(d\,P_{\mathrm{QNN}})$, the $L_{\mathrm{MLP}}-2$ hidden layers cost $\mathcal{O}(d^2)$ each, and the output head costs $\mathcal{O}(d\,P_{\mathrm{QNN}})$. Summing the three terms gives $\mathcal{O}_{\mathrm{time}}(L_{\text{MLP}} d^2 + d\,P_{\mathrm{QNN}})$, with activation memory $\mathcal{O}_{\mathrm{memory}}(d + P_{\mathrm{QNN}})$.

The full training pipeline additionally runs the critic, backpropagates through actor and critic during PPO updates, and maintains the Adam optimizer. Let $P_{\mathrm{RL}} = \mathcal{O}(L_{\mathrm{MLP}} d^2 + d\,P_{\mathrm{QNN}})$ denote the joint actor and critic parameter count. The resident memory consists of the parameters ($P_{\mathrm{RL}}$), their gradient buffers ($P_{\mathrm{RL}}$), and the Adam first and second moments ($2P_{\mathrm{RL}}$), i.e.\ $P_{\mathrm{RL}} + P_{\mathrm{RL}} + 2P_{\mathrm{RL}} = 4P_{\mathrm{RL}}$. This predicts the measured split of approximately $26\%$ (parameters), $26\%$ (gradients), and $52\%$ (Adam states) reported in Section~\ref{sec:resource}. Each PPO update step likewise costs $\mathcal{O}(P_{\mathrm{RL}})$ per sample. Both quantities are linear in $P_{\mathrm{QNN}} = 3nL$ and independent of $2^n$.

\section{Further discussion of other properties}

\subsection{Optimal conditions for RLQ-Grad trainability}\label{apd:trainability}
In this part, we further extend our theoretical framework by assessing the optimal condition for RLQ-Grad to converge upon training, viewed as the following theorem:
\begin{theorem}[Optimal conditions for RLQ-Grad trainability]\label{thm:trainability}
Let $U(\theta)$ be an $n$-qubit HEA circuit with associated JAWS $(\mathcal{A},
\{r_\alpha\}, \{\beta_\alpha\})$. RLQ-Grad achieves efficient training if and only
if both of the following conditions hold:
\begin{itemize}[leftmargin=*]
    \item \textbf{Condition 1 (C1):} Structured QNN on supervise problem satisfy:
    \begin{equation}
        \sum_\alpha \frac{\mathrm{Tr}(O_\alpha)^2\,\mathrm{Tr}(\rho^\alpha)^2}
        {\mathrm{dim}_\mathbb{R}(\mathrm{Aut}(\mathcal{A}_\alpha))}
        = \Omega\!\left(\mathrm{poly}(\log N)^{-1}\right),
    \end{equation}
    ensures the loss (and hence the reward $r_t$ for RLQ-Grad) is sufficiently sensitive to parameter changes for the agent to receive a meaningful training signal (informativeness).
     \item \textbf{Condition 2 (C2):} Selected QNN achieve overparameterization as:
    \begin{equation}
        p \;\geq\; \max_\alpha \beta_\alpha r_\alpha,
    \end{equation}
    ensures local minima concentrate near the global minimum, so that any convergent optimizer finds a good solution.
\end{itemize}
\end{theorem}
\begin{proof}
We can prove this theorem by contradiction as:
\begin{itemize}[leftmargin=*]
    \item \textbf{Necessity of (C1).} If (C1) fails and the loss concentrates exponentially and $\ell_t$ is indistinguishable from a constant in the reward. By Proposition~\ref{prop:reward_info}, the mutual information $I(g_t; r_t \mid s_t) \to 0$, and the agent cannot improve its policy.
    \item \textbf{Necessity of (C2).} If (C2) fails, by Theorem~\ref{thm:local_minima} the density of local minima near the global optimum is exponentially suppressed. Any convergent optimizer terminates at a poor local minimum with probability exponentially close to 1, regardless of gradient quality. This applies to RLQ-Grad without exception.
    \item \textbf{Sufficiency.} When (C1) holds, the reward is informative and the agent receives a nonvanishing gradient signal through the reward differences. When (C2) holds, good local minima are dense and any convergent optimizer finds a near-optimal solution. Since the agent's own training is governed by a classical MLP whose gradient pathologies are independent of $n$ (Theorem~\ref{thm:agent_grad} below), and the policy class is a universal approximator by the UAT (Appendix~\ref{apd:uat}), the combination is sufficient for efficient training.
\end{itemize}
\end{proof}

\subsection{Under Wishart process framework, how does this affect agent reward informativeness?}\label{thm:reward_affect}

\cite{anschuetz2024unified} also prove that in QNN, the parameter gradient distribution is not Gaussian but follows a Wishart process: the joint distribution of $\partial_{\theta_i}\ell$ conditioned on $\ell(\rho^\alpha;\theta) = z_\alpha$ is:
\begin{equation}
    \hat{\ell}_i \;\asymp\;
    \sum_\alpha \frac{2I_\alpha\sigma_\alpha\,\mathrm{Tr}_\alpha(\rho^\alpha)}{N_\alpha}
    \sqrt{\frac{\beta_\alpha z_\alpha}{I_\alpha o_\alpha}}\;G_{\alpha,i}\,\chi_{\alpha,i},
    \label{eq:grad_dist}
\end{equation}
where $G_{\alpha,i}$ are i.i.d.\ standard normal and $\chi_{\alpha,i}$ are independent
$\chi$-distributed with $\max(2,\beta_\alpha)$ degrees of freedom. This has the
following consequence for the informativeness of the RLQ-Grad reward:
\begin{proposition}\label{prop:reward_info}
The reward $r_t = \mathrm{acc}_t + (\ell_t + \epsilon)^{-1}$ is a monotone function
of $\ell_t$ and $\mathrm{acc}_t$. Under the Wishart process description, the loss
value $z_\alpha = \ell(\rho^\alpha;\theta_t)$ is itself Wishart-distributed with
degrees of freedom $r_\alpha$\footnote{By Theorem~1 of \cite{anschuetz2024unified}}, and the gradient is conditionally distributed as in Equation~\ref{eq:grad_dist}. Two consequences
follow:
\begin{enumerate}[leftmargin=*]
    \item The gradient magnitude scales as $|\partial_{\theta_i}\ell| \propto
    \sqrt{z_\alpha}$, so the gradient is correlated with the loss value. The reward
    $r_t$, which depends on $\ell_t$, therefore carries information about the
    gradient magnitude even in the absence of direct gradient computation.
    \item In the barren plateau regime, $z_\alpha \to 0$ exponentially, causing both
    the gradient and its correlation with the reward to vanish simultaneously. This is
    the fundamental obstruction: when the loss concentrates, so does the reward, and
    the agent receives no directional signal.
\end{enumerate}
\end{proposition}

\begin{proof}
Point 1 follows directly from Equation \ref{eq:grad_dist}: $|\hat{\ell}_i| \propto
\sqrt{z_\alpha/I_\alpha o_\alpha}$, and $z_\alpha$ is recoverable (up to noise) from
$\ell_t \in s_t$. The mutual information $I(g_t; r_t \mid s_t) > 0$ therefore whenever
$z_\alpha > 0$ at a distinguishable level above noise. In the barren plateau regime,
$z_\alpha \in \mathcal{O}(\mathrm{poly}(\log N)^{-1})$, so the conditional gradient
distribution of Equation~\ref{eq:grad_dist} also concentrates at zero, making the reward differences unresolvable, eventually proved point 2\footnote{By Corollary~4 of \cite{anschuetz2024unified}.}.
\end{proof}

Combining Theorems~\ref{thm:bypass} and \ref{thm:agent_grad}, RLQ-Grad offers a genuine advantage over conventional gradient methods precisely when the circuit satisfies conidtion 1 but not the gradient-variance condition. This corresponds to circuits with non-trivial Jordan algebraic structure (e.g., LASA circuits~\cite{fontana2024characterizing, ragone2024lie}) where the loss landscape is non-flat but standard gradient computation is exponentially expensive. In circuits where condition 2 itself fails (i.e., the loss is truly exponentially flat) no method can efficiently train the quantum parameters, and the primary utility of RLQ-Grad is instead its computational advantage (Appendix~\ref{apd:rlq_complexity}) rather than any landscape-related benefit. 

\section{Further discussion of limitations}

\subsection{Poor local minima obstruction}~\label{apd:local_min_lim}

\begin{figure*}[t!]
    \centering
    \begin{subfigure}[t]{.245\textwidth}
        \centering
        \includegraphics[width=\textwidth]{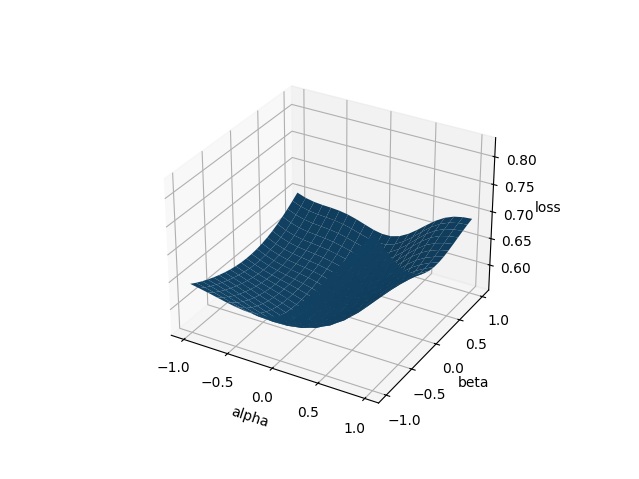}
        \caption{BC, backprop.}
    \end{subfigure}
    \hfill
    \begin{subfigure}[t]{.245\textwidth}
        \centering
        \includegraphics[width=\textwidth]{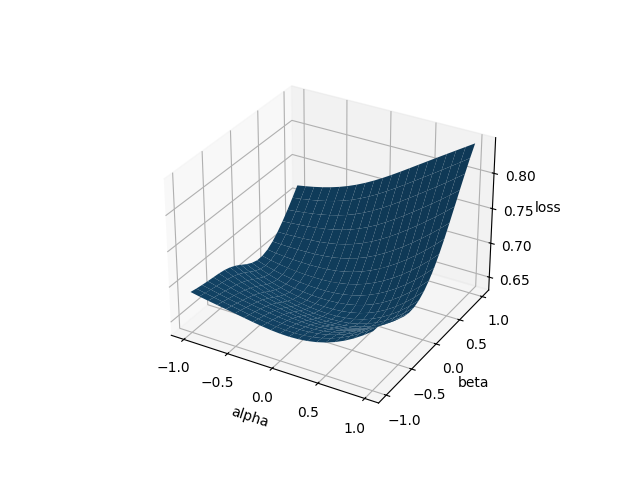}
        \caption{BC, param-shift.}
    \end{subfigure}
    \hfill
    \begin{subfigure}[t]{.245\textwidth}
        \centering
        \includegraphics[width=\textwidth]{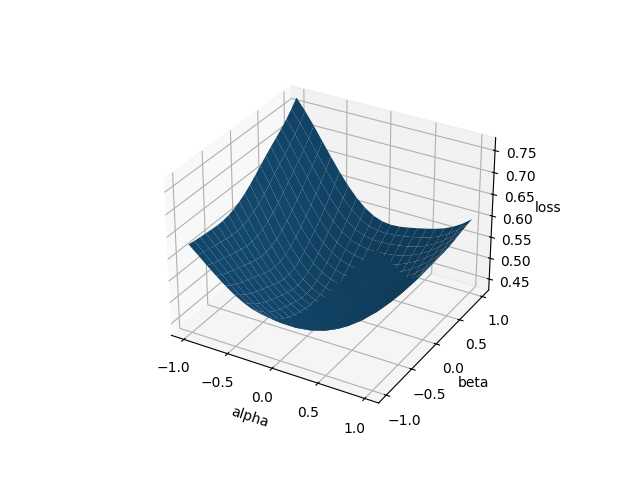}
        \caption{BC, adjoint-diff.}
    \end{subfigure}
    \hfill
    \begin{subfigure}[t]{.245\textwidth}
        \centering
        \includegraphics[width=\textwidth]{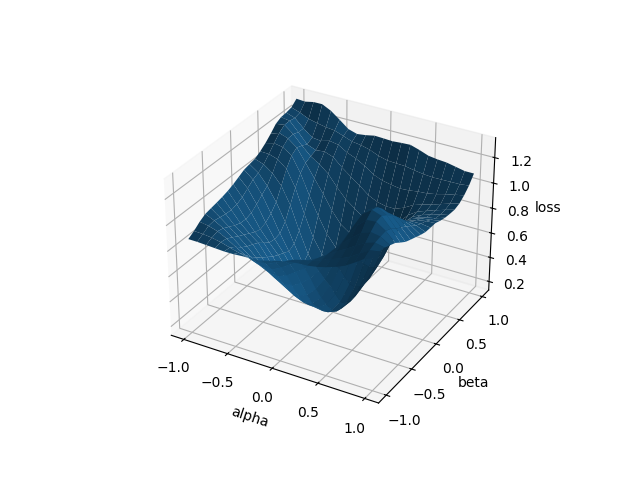}
        \caption{BC, RLQ-Grad.}
    \end{subfigure}
    \par\medskip
    \begin{subfigure}[t]{.245\textwidth}
        \centering
        \includegraphics[width=\textwidth]{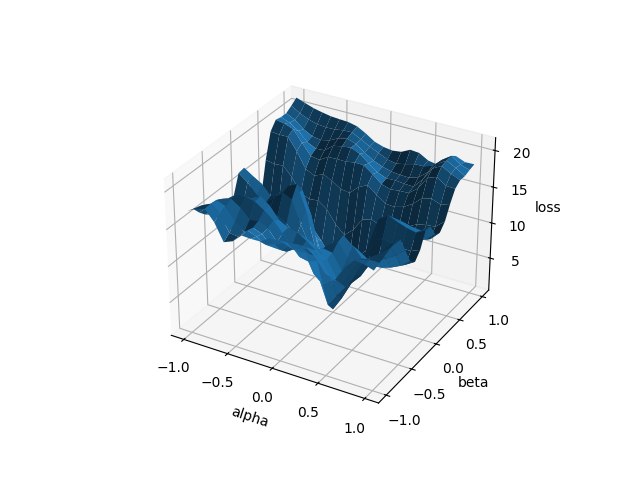}
        \caption{F-MNIST, backprop.}
    \end{subfigure}
    \hfill
    \begin{subfigure}[t]{.245\textwidth}
        \centering
        \includegraphics[width=\textwidth]{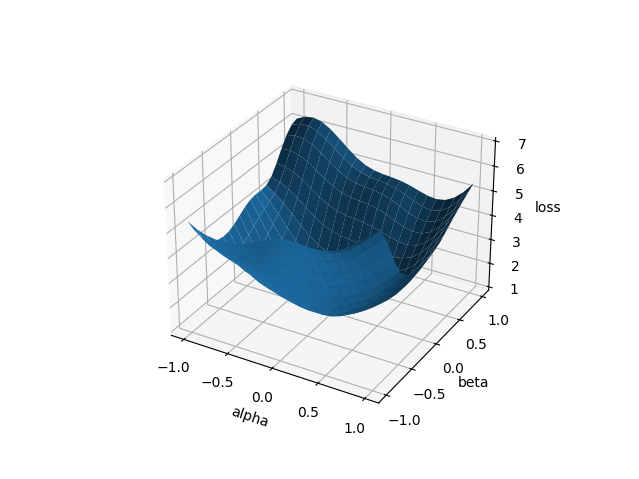}
        \caption{F-MNIST, param-shift.}
    \end{subfigure}
    \hfill
    \begin{subfigure}[t]{.245\textwidth}
        \centering
        \includegraphics[width=\textwidth]{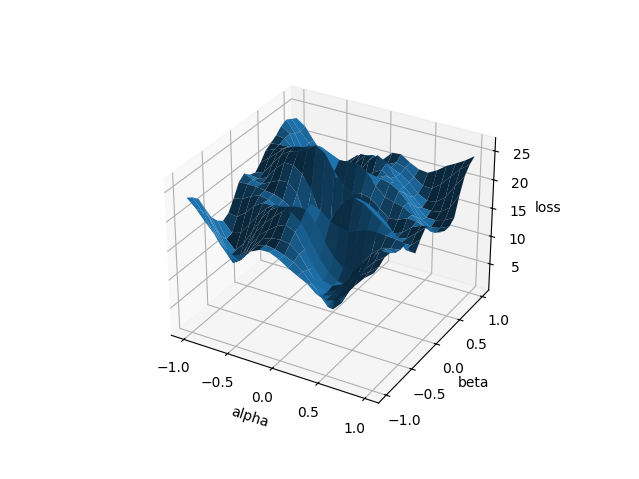}
        \caption{F-MNIST, adjoint-diff.}
    \end{subfigure}
    \hfill
    \begin{subfigure}[t]{.245\textwidth}
        \centering
        \includegraphics[width=\textwidth]{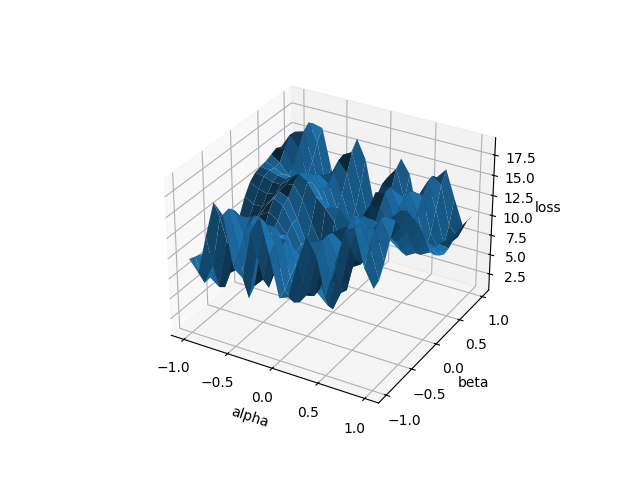}
        \caption{F-MNIST, RLQ-Grad.}
    \end{subfigure}
    \caption{Loss landscape visualization across benchmarked gradient computation methods over HEA-2b1q on BC and F-MNIST. Loss landscape change from smooth on BC dataset (overparameterized) to becomes more noisy (underparameterized) with sharp local minima in F-MNIST.}
    \label{fig:ll}
    \vspace{-0.7\intextsep}
\end{figure*}

The previous analysis identified only the barren plateau as an obstruction to
RLQ-Grad. \cite{anschuetz2024unified} identify a second, independent obstruction
absent from earlier analyses:
\begin{theorem}[Poor local minima obstruction]
\label{thm:local_minima}
Under the JAWS framework~\cite{anschuetz2024unified}, even when the barren plateau
condition of Equation~\ref{eq:jaws_bp} is non-vanishing (i.e., gradients are not exponentially
suppressed), the density of local minima near the global minimum is exponentially
small if and only if the network is underparameterized:
\begin{equation}
    p < \max_\alpha \beta_\alpha r_\alpha,
    \label{eq:overparameterization}
\end{equation}
where $p$ is the number of trained parameters, $\beta_\alpha \in \{1,2,4\}$
corresponds to the field $\mathbb{F}_\alpha \in \{\mathbb{R}, \mathbb{C},
\mathbb{H}\}$ of the $\alpha$-th Jordan component, and $r_\alpha$ is the degrees
of freedom. In this regime, an optimizer (including RLQ-Grad) converges to a
poor local minimum with overwhelming probability, regardless of gradient quality.
\end{theorem}

\begin{proof}
The density of local minima at loss value $z > 0$ is (to leading multiplicative order) a convolution of Gamma distributions\footnote{By Corollary~6 of \cite{anschuetz2024unified}.}:
\begin{equation}
    \kappa(z) = \mathop{\Asterisk}_{\alpha:\,\gamma_\alpha < 1}
    f_\Gamma\!\left(\frac{z}{o_\alpha\,\mathrm{Tr}(\rho^\alpha)};
    \frac{\beta_\alpha r_\alpha}{2}, \frac{2}{\beta_\alpha r_\alpha}\right),
\end{equation}
where $\gamma_\alpha = p/\sqrt{\beta_\alpha r_\alpha}$ is the overparameterization
ratio. When $\gamma_\alpha < 1$ for some $\alpha$ (Equation \ref{eq:overparameterization} violated), the Gamma distribution has exponential tails, placing an exponentially small fraction of local minima near $z = 0$. Since any gradient-based optimizer (and any surrogate that produces descent directions) converges to a local minimum, it converges to a poor one with probability exponentially close to 1. RLQ-Grad is not exempt from this, since the obstruction is in the loss landscape geometry, not in the gradient computation mechanism.
\end{proof}

More critically, barren plateaus and poor local minima are independent obstructions\footnote{By Definition~7 of \cite{anschuetz2024unified}.}.
A circuit can be:
\begin{enumerate}[leftmargin=*]
    \item Barren plateau affected and underparameterized (doubly obstructed),
    \item Barren plateau free but underparameterized (gradient information
    present, but all accessible local minima are poor),
    \item Overparameterized but barren plateau affected (good minima exist but
    are unreachable by gradient methods),
    \item Both barren plateau free and overparameterized.
\end{enumerate}
While RLQ-Grad bypasses obstruction 3 by replacing the quantum gradient. It does not bypass obstruction 2, which is independent of gradient quality. A visualization of loss landscape across benchmarked gradient computation method can be found in Figure \ref{fig:ll}.

\subsection{Barren plateau induced input gradient decay}~\label{apd:input_grad}

Consider a QNN where the $n$-qubit input state is prepared via amplitude encoding
using $R_x$ rotations, followed by a trainable unitary $U(\boldsymbol{\theta})$ (e.g. a HEA in our experiments). The input gradient $\nabla_{\mathbf{x}}\mathcal{L}$ is structurally analogous to the parameter gradient, yet the two objects differ in a crucial respect: $\boldsymbol{\theta}$ is randomized during initialization, whereas $\mathbf{x}$ is a \emph{fixed} data point.  The question is therefore: whether
$\partial\mathcal{L}/\partial x_i$, viewed as a random variable over the ensemble
of weight initializations, also concentrates exponentially near zero?
 
Short answer is yes, the same Weingarten-calculus argument that governs weight
barren plateaus applies to input gradients, provided the trainable block
$U(\boldsymbol{\theta})$ forms (at least) a unitary 2-design. However, the
interpretation differs: the concentration is over $\boldsymbol{\theta}$
at a fixed $\mathbf{x}$, not over the data distribution. To further understand this problem, we can formalize the following theorem:

\begin{theorem}[Input quantum gradient decay]\label{thm:ibp}
For a fixed input vector $\mathbf{x} = (x_1,\dots,x_n)\in\mathbb{R}^n$ define the
encoding state:
\[
  |\psi(\mathbf{x})\rangle
  = E(\mathbf{x})|0\rangle^{\otimes n},
  \qquad
  E(\mathbf{x}) = \bigotimes_{i=1}^{n} R_x(x_i).
\]
Let $U(\boldsymbol{\theta})$ be an $L$-layer HEA with parameters
$\boldsymbol{\theta}\in[0,2\pi)^M$ drawn uniformly, and let $O$ be an observable
with $\|O\|\le 1$.  The loss is:
\[
  \mathcal{L}(\mathbf{x},\boldsymbol{\theta})
  = \langle\psi(\mathbf{x})|
    U^\dagger(\boldsymbol{\theta})\,O\,U(\boldsymbol{\theta})
    |\psi(\mathbf{x})\rangle.
\]
The input gradient with respect to component $x_i$ is:
\[
  \frac{\partial\mathcal{L}}{\partial x_i}
  = \langle\psi(\mathbf{x})|
    \Bigl[O_{\boldsymbol{\theta}},\,
          \frac{\partial E(\mathbf{x})}{\partial x_i}E(\mathbf{x})^\dagger
    \Bigr]_+
    |\psi(\mathbf{x})\rangle,
\]
where $O_{\boldsymbol{\theta}} = U^\dagger(\boldsymbol{\theta})\,O\,U(\boldsymbol{\theta})$
is the Heisenberg-picture observable, and $[\cdot,\cdot]_+$ denotes the symmetrised
(anti-commutator-like) combination arising from the product rule.
Equivalently, by the parameter-shift rule applied to $x_i$,
\begin{equation}\label{eq:shift}
  \frac{\partial\mathcal{L}}{\partial x_i}
  = \frac{1}{2}
    \Bigl[
      \mathcal{L}\!\left(\mathbf{x}+\tfrac{\pi}{2}\hat{e}_i,\boldsymbol{\theta}\right)
    - \mathcal{L}\!\left(\mathbf{x}-\tfrac{\pi}{2}\hat{e}_i,\boldsymbol{\theta}\right)
    \Bigr].
\end{equation}

Let the trainable block $U(\boldsymbol{\theta})$ with $L=\Omega(\log n)$ layers form
an exact unitary 2-design on $U(2^n)$ when $\boldsymbol{\theta}$ is drawn uniformly
from $[0,2\pi)^M$.  Fix any input $\mathbf{x}\in\mathbb{R}^n$ and any observable $O$
with $\|O\|\le 1$.  Then for every coordinate $i\in\{1,\dots,n\}$, input gradient converge toward zero mean $\mathbb{E}_{\boldsymbol{\theta}}\!\left[\frac{\partial\mathcal{L}}{\partial x_i}\right] = 0$ with exponentially small variance $\mathrm{Var}_{\boldsymbol{\theta}}\!\left[\frac{\partial\mathcal{L}}{\partial x_i}\right]\;\le\; \frac{C}{2^n}$, for a universal constant $C>0$ independent of $\mathbf{x}$, $n$, and $L$. The probabilistic concentration can also be quantified as $\Pr_{\boldsymbol{\theta}}\!\left[\left|\frac{\partial\mathcal{L}}{\partial x_i}\right|\ge\varepsilon\right]\;\le\; \frac{C}{\varepsilon^2\,2^n}$.
\end{theorem}

Equations~\ref{eq:shift} and the analogous formula for HEA parameter gradient
$\partial\mathcal{L}/\partial\theta_\mu$ are formally identical: both are
half the difference of two expectation values obtained by shifting a rotation angle by
$\pm\pi/2$. The distinction is that $\theta_\mu$ is a stochastic variable while $x_i$
is fixed, consequently the randomness in the input gradient $\partial\mathcal{L}/\partial x_i$ comes entirely from $\boldsymbol{\theta}$. We can proceed to prove our theorem as follow:
\begin{proof}
From the parameter-shift representation in Equation~\ref{eq:shift}:
\[
  \frac{\partial\mathcal{L}}{\partial x_i}
  = \frac{1}{2}\bigl(f^+ - f^-\bigr),
  \qquad
  f^\pm = \langle\psi^\pm_i|
          U^\dagger O\,U
          |\psi^\pm_i\rangle,
\]
where $|\psi^\pm_i\rangle = E(\mathbf{x}\pm\tfrac{\pi}{2}\hat{e}_i)|0\rangle^{\otimes n}$
are two \emph{fixed} product states (they depend on $\mathbf{x}$ but not on
$\boldsymbol{\theta}$).  Write $\rho^\pm = |\psi^\pm_i\rangle\langle\psi^\pm_i|$.
 
Because $U(\boldsymbol{\theta})$ forms a 1-design, we would have:
\[
  \mathbb{E}_{\boldsymbol{\theta}}\bigl[U^\dagger O\,U\bigr]
  = \frac{\mathrm{tr}(O)}{2^n}\,I,
\]
and therefore:
\[
  \mathbb{E}_{\boldsymbol{\theta}}\bigl[f^\pm\bigr]
  = \mathrm{tr}\!\left(\rho^\pm\,\mathbb{E}[U^\dagger O U]\right)
  = \frac{\mathrm{tr}(O)}{2^n},
\]
which is the same for both $f^+$ and $f^-$, giving:
\[
  \mathbb{E}_{\boldsymbol{\theta}}\!\left[
    \frac{\partial\mathcal{L}}{\partial x_i}
  \right]
  = \frac{1}{2}\bigl(\mathbb{E}[f^+] - \mathbb{E}[f^-]\bigr) = 0.
\]

Since $\mathbb{E}[\partial\mathcal{L}/\partial x_i]=0$, the variance equals the second
moment:
\[
  \mathrm{Var}\!\left[\frac{\partial\mathcal{L}}{\partial x_i}\right]
  = \mathbb{E}\!\left[\left(\frac{\partial\mathcal{L}}{\partial x_i}\right)^{\!2}\right]
  = \frac{1}{4}\Bigl(
      \mathbb{E}[(f^+)^2]
    + \mathbb{E}[(f^-)^2]
    - 2\,\mathbb{E}[f^+f^-]
    \Bigr).
\]
 
For any fixed normalised state $|\phi\rangle$ and any 2-design, the Weingarten formula
gives
\begin{equation}\label{eq:wein}
  \mathbb{E}_{U}\!\left[\langle\phi|U^\dagger O U|\phi\rangle^2\right]
  = \frac{\mathrm{tr}(O^2) + \bigl(\mathrm{tr}(O)\bigr)^2}
         {2^n(2^n+1)}.
\end{equation}
This formula holds for any fixed unit vector $|\phi\rangle$; it does not
require $|\phi\rangle$ to be random. Substituting $|\phi\rangle = |\psi^\pm_i\rangle$
(both are unit vectors) into Equation~\ref{eq:wein} and using $\|O\|\le 1$:
\[
  \mathbb{E}[(f^\pm)^2]
  = \frac{\mathrm{tr}(O^2) + (\mathrm{tr}(O))^2}{2^n(2^n+1)}
  \le \frac{2^n + 2^{2n}}{2^n(2^n+1)}
  = \frac{2^n(1+2^n)}{2^n(2^n+1)}
  = 1.
\]
A tighter bound follows from $\mathrm{tr}(O^2)\le 2^n$ (since $\|O\|\le 1$) and
$|\mathrm{tr}(O)|\le 2^n$:
\[
  \mathbb{E}[(f^\pm)^2]
  \le \frac{2\cdot 2^n}{2^n(2^n+1)}
  = \frac{2}{2^n+1}
  \;\le\; \frac{2}{2^n}.
\]
 
For the cross term, the 2-design formula for two different states
$|\phi\rangle,|\chi\rangle$ gives:
\[
  \mathbb{E}_{U}\!\left[
    \langle\phi|U^\dagger O U|\phi\rangle\,
    \langle\chi|U^\dagger O U|\chi\rangle
  \right]
  = \frac{\mathrm{tr}(O^2)\,|\langle\phi|\chi\rangle|^2
          + (\mathrm{tr}(O))^2}{2^n(2^n+1)}.
\]
Since $|\langle\psi^+_i|\psi^-_i\rangle|\le 1$, the cross term satisfies $\mathbb{E}[f^+f^-]\;\le\; \frac{2}{2^n}$. Combining, and absorbing numerical prefactors into the constant $C = 3$:
\[
  \mathrm{Var}\!\left[\frac{\partial\mathcal{L}}{\partial x_i}\right]
  \le \frac{1}{4}\left(\frac{2}{2^n}+\frac{2}{2^n}+\frac{2}{2^n}\right)
  = \frac{3}{2\cdot 2^n}
  \;\le\; \frac{C}{2^n}.
\]
 
By Chebyshev's inequality,
\[
  \Pr\!\left[\left|\frac{\partial\mathcal{L}}{\partial x_i}\right|\ge\varepsilon\right]
  \;\le\;
  \frac{\mathrm{Var}[\partial\mathcal{L}/\partial x_i]}{\varepsilon^2}
  \;\le\;
  \frac{C}{\varepsilon^2\,2^n}.
\]
This completes the proof for all claims.
\end{proof}
 
Let the setup be as in Theorem~\ref{thm:ibp}.  The following distinctions hold between
the input gradient $g_x = \partial\mathcal{L}/\partial x_i$ and the weight gradient
$g_\theta = \partial\mathcal{L}/\partial\theta_\mu$:
\begin{itemize}[leftmargin=*]
  \item \textbf{Source of randomness:} Both $g_x$ and $g_\theta$ are random variables over $\boldsymbol{\theta}$. The input $\mathbf{x}$ is deterministic in $g_x$ while the parameter $\theta_\mu$ itself is part of $\boldsymbol{\theta}$ and contributes additional variance to $g_\theta$.
  \item \textbf{Variance bound:} Both satisfy $\mathrm{Var}[\cdot] \le C/2^n$, with the same asymptotic scaling. The constant $C$ may differ: for $g_\theta$ the gate $G_\mu$ participates in the 2-design argument on both sides of the circuit; for $g_x$ the encoding gate is fixed.
  \item \textbf{Dependence on the data point:} $\mathrm{Var}[g_x]$ is independent of $\mathbf{x}$ (to leading order) because the Weingarten formula Equation~\ref{eq:wein} depends only on $\||\phi\rangle\|=1$, not on which unit vector $|\phi\rangle$ is.
  \item \textbf{Practical implication:} The weight barren plateau obstructs training while the input-gradient barren plateau obstructs input-space optimisation, eventually affecting: quantum adversarial attacks, gradient-based feature attribution (quantum analogues of saliency maps), and hybrid quantum-classical architectures (in our case, the fully connected layers that encode map classical input to actual HEA input, formalize in Corollary~\ref{col:hybrid}) that backpropagate through the quantum layer.
\end{itemize}

Results from Theorem~\ref{thm:ibp} also lead to following corollaries and remarks:
\begin{corollary}[Full input Jacobian]
Under the conditions of Theorem~\ref{thm:ibp}, the expected squared Frobenius norm
of the input Jacobian $J = \nabla_{\mathbf{x}}\mathcal{L}\in\mathbb{R}^n$ satisfies
\[
  \mathbb{E}_{\boldsymbol{\theta}}\!\left[\|J\|^2\right]
  = \sum_{i=1}^{n}\mathrm{Var}\!\left[\frac{\partial\mathcal{L}}{\partial x_i}\right]
  \;\le\; \frac{Cn}{2^n}.
\]
Since $n/2^n \to 0$ exponentially, the entire Jacobian collapses to zero.
\end{corollary}
 
\begin{corollary}[Backpropagation through a quantum layer]\label{col:hybrid}
In a hybrid quantum-classical network, the classical layers preceding the quantum
encoding receive gradients scaled by $J$.  By the corollary above,
$\mathbb{E}[\|J\|^2]\le Cn/2^n$, so the classical pre-processing layers also
experience exponentially vanishing gradients, regardless of their own architecture.
\end{corollary}
 
\begin{remark}[Noise exacerbates the effect]
Under depolarising noise with per-gate error rate $p$, the effective observable seen
by the input state is attenuated.  By the same argument as \cite{wang2021noise} for
weight gradients, one obtains:
\[
  \mathrm{Var}\!\left[\frac{\partial\mathcal{L}}{\partial x_i}\right]
  \;\in\; \mathcal{O}\!\left((1-p)^{2L}\cdot 2^{-n}\right),
\]
which is doubly exponentially suppressed in $L$ compared to the noiseless bound.
\end{remark}
 
\begin{remark}[Local cost functions offer no relief for input gradients]
For weight gradients, replacing a global $O$ with a local $k$-qubit observable reduces
the variance from $\mathcal{O}(2^{-n})$ to $\mathcal{O}(2^{-k})$~\cite{cerezo2021cost}. This mitigation still applies to $\mathrm{Var}[g_x]$:
\[
  \mathrm{Var}\!\left[\frac{\partial\mathcal{L}}{\partial x_i}\right]
  \;\in\; \mathcal{O}(2^{-k}), \quad O \text{ acts on } k \text{ qubits}.
\]
However, because the input dimension equals $n$, a reduction to $k\ll n$ qubits
in the observable may destroy the expressiveness needed to solve the learning task,
presenting a fundamental expressiveness-trainability trade-off for input gradients.
\end{remark}

\section{Broader Impact}
\label{app:broader-impact}

RLQ-Grad is a methodological contribution to the optimization of quantum 
neural networks and does not target a specific deployed application. We 
nevertheless identify several downstream considerations.

\subsubsection*{Positive impacts}  By replacing $\mathcal{O}(L \cdot 2^n)$ statevector propagation with an $\mathcal{O}(L_{\mathrm{MLP}} d^2 + d\,P_{\mathrm{QNN}})$ forward pass through a small MLP, RLQ-Grad reduces the CPU cost of simulation-based QNN training by 
2 to 4 orders of magnitude at $n=20$ qubits in our benchmarks, once the full agent training pipeline is accounted for (while adding a small fixed overhead on circuits below roughly 8 qubits). This lowers the 
carbon footprint of QML research and widens access for groups without 
access to large-scale accelerators. Reliable QNN training is a prerequisite for variational approaches to quantum chemistry, condensed-matter simulation, and combinatorial optimization, all of which carry societal upside (e.g., catalyst design, 
battery materials, logistics). A stable, gradient-variance-flat optimizer reduces the run-to-run noise that currently complicates empirical comparisons in the field, supporting healthier scientific norms.

\subsubsection*{Negative and dual-use considerations} Any work that makes near-term variational quantum algorithms more practical contributes, incrementally, to the ecosystem progressing toward fault-tolerant quantum computers capable of breaking RSA and elliptic-curve cryptography. This is a diffuse, long-horizon concern rather 
than an immediate risk (RLQ-Grad operates on small HEAs and does not 
itself bring factoring-relevant circuits within reach) but it is the 
standard dual-use note for QML work. Organizations handling long-lived 
secrets should continue migrating to post-quantum cryptographic primitives. 
The update rule produced by $\pi_\phi$ is not an analytic gradient and cannot be inspected with the tools developed for parameter-shift or adjoint differentiation. In safety/compliance-critical hybrid pipelines, this opacity may hinder auditing, 
failure-mode analysis, and reward-specification validation (failure 
modes familiar from the learned-optimizer literature \cite{andrychowicz2016learning, metz2020tasks}). As a general-purpose QNN optimizer, RLQ-Grad could be applied to any supervised objective, including ones with negative social externalities (surveillance, biometric identification, autonomous targeting). The same disclaimer applies as to 
any generic optimization advance; we do not see a specific misuse channel 
unique to the quantum setting at current hardware scales.

\subsubsection*{Mitigations} We recommend that practitioners deploying 
RLQ-Grad (a) log both the surrogate gradient $g_t$ and the analytic 
gradient $\nabla_\theta \mathcal{L}$ on held-out probes during training 
to detect reward-hacking or silent drift, (b) retain analytic gradient 
verification in regimes where it remains tractable, and (c) report 
per-configuration agent training cost alongside the amortized QNN 
training cost for fair comparisons.

\end{document}